\documentclass{article}

\PassOptionsToPackage{numbers, compress}{natbib}
\usepackage[preprint]{neurips_2026}

\usepackage[utf8]{inputenc} 
\usepackage[T1]{fontenc}    
\usepackage{hyperref}       
\usepackage{url}            
\usepackage{booktabs}       
\usepackage{amsfonts}       
\usepackage{nicefrac}       
\usepackage{microtype}      
\usepackage{xcolor}         
\usepackage{graphicx}
\usepackage{amsmath,amsthm,stackrel, amssymb, mathtools}
\usepackage{bbm}
\usepackage{algorithm}
\usepackage{algpseudocode}
\usepackage{float}
\usepackage{wrapfig}

\newcommand{\R}{\mathbb{R}}
\newcommand{\N}{\mathbb{N}}
\newcommand{\prob}{\mathbb{P}}

\newcommand{\E}{\mathbb{E}}

\newcommand{\pto}{\xrightarrow{p}}
\newcommand{\dto}{\xrightarrow{d}}

\usepackage{float}
\newcommand{\expMech}{\texttt{ExpMech}}
\newcommand{\blb}{\texttt{BLBquant}}
\newcommand{\privSub}{\texttt{PrivSub}}
\newcommand{\privBoot}{\texttt{PrivBoot}}

\newcommand{\dataset}{\boldsymbol{\omega}}
\newcommand{\datadom}{\Omega}

\DeclareMathOperator{\Var}{Var}

\theoremstyle{plain}
\newtheorem{theorem}{Theorem}[section]

\newtheorem{lemma}[theorem]{Lemma}
\newtheorem{corollary}[theorem]{Corollary}
\theoremstyle{definition}
\newtheorem{definition}[theorem]{Definition}

\theoremstyle{remark}
\newtheorem{remark}[theorem]{Remark}
\newtheorem{claim}[theorem]{Claim}

\title{Differentially Private Nonparametric Confidence Intervals Under Minimal Distributional Assumptions}

\author{%
  Tomer Shoham\thanks{Tomer.Shoham@mail.huji.ac.il} \quad
  Moshe Shenfeld \quad
  Noa Velner-Harris \quad
  Katrina Ligett \\
  \\
  School of Computer Science and Engineering\\
  The Hebrew University of Jerusalem\\
  Jerusalem, Israel 
}

\begin{document}

\maketitle

\begin{abstract}
  We consider the problem of constructing differentially private nonparametric confidence intervals (CIs)  for an arbitrary quantity using resampling. A growing body of work has adapted resampling ideas to the private setting, including private bootstrap methods \cite{brawner2018bootstrap, wang2025differentially,dette2025gaussian} and BLB-based subsample-and-aggregate approaches \cite{covington2025unbiased, chadha2024resampling}. However, existing methods typically rely on strong assumptions, such as asymptotic normality, or are tied to specific privacy mechanisms such as noise addition, and can be impractical in finite-sample regimes. We address these problems by introducing a simple, general framework that can convert any differentially private estimator satisfying mild conditions into a differentially private nonparametric CI for arbitrary target quantities. Our method repeatedly subsamples the data, applies the private estimator to each subset, and post-processes the resulting empirical CDF into a CI. The framework is black-box, and does not require a specific limiting distribution. We prove that the empirical CDF induced by our procedure converges to the sampling distribution of the private statistic, which implies that the resulting CI is asymptotically valid and tight, and provide heuristic guidance for choosing the hyperparameters. Empirically, our method outperforms competing general approaches, especially for non-smooth functionals and more challenging distributions.
\end{abstract}

\section{Introduction and related work}\label{sec:intro}
 
Confidence intervals (CIs) are a central tool in statistical inference: they quantify uncertainty around an estimator and support downstream tasks such as hypothesis testing and scientific reporting. A CI is a data-dependent subset of the parameter space that aims to contain the true value with a prespecified probability (the {\em confidence level}, typically 90–99\%), over independent sampling from the same data-generating process. The complement is the significance level $\alpha$: the tolerated probability that the returned set misses the truth. Note that since the dataset is random, the output set is random as well. We refer to the probability that the CI contains the true value as its {\em coverage}, and consider its accuracy to be the gap between the coverage and the desired confidence level. This gap is not symmetric: CIs that under-cover fail to provide the advertised inferential guarantee, and thus are invalid, whereas over-covering CIs are valid but conservative, typically at the cost of wider intervals.

\begin{figure*}[t!]
\vspace{.3in}
\centering\includegraphics[width=\linewidth]{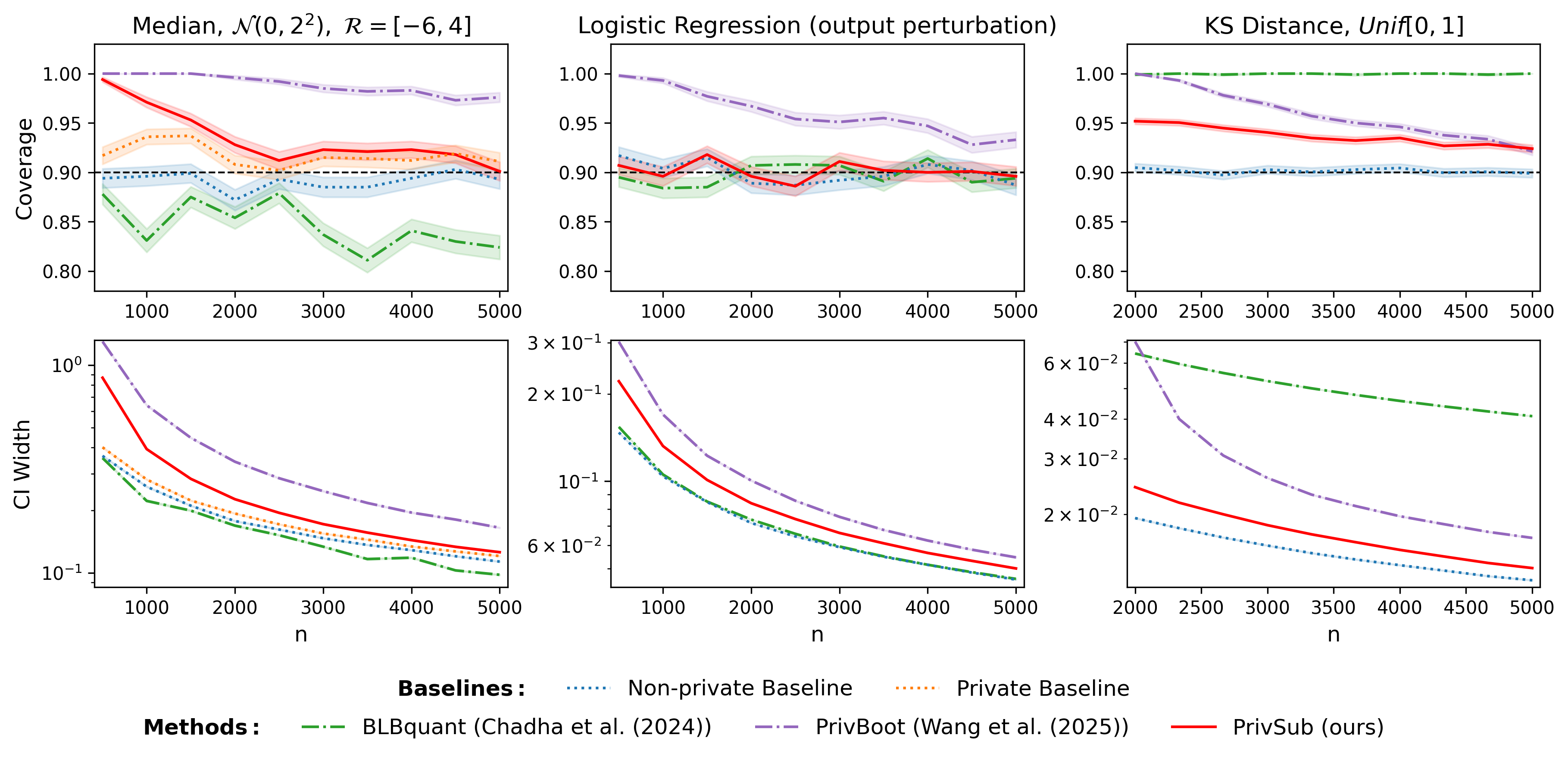}
\vspace{-0.5cm}
\caption{
Finite-sample comparison of \privSub\ with existing general nonparametric DP CI methods: \blb\ \citep{chadha2024resampling} and \privBoot\ \citep{wang2025differentially}. Across median estimation, logistic regression slope estimation, and KS-distance estimation, \privSub\ produces valid $0.9$-CIs with competitive, and often smaller, widths. We also report a median-specific exponential-mechanism baseline (\expMech; \citealp{drechsler2022nonparametric}) and the corresponding non-private baseline: bootstrap for the median and logistic regression, and subsampling for the KS statistic. For the median and KS experiments, all private methods are $5$-pure DP, except \privBoot, which is $(\varepsilon=5,\delta=10^{-6})$-DP; for logistic regression, all private methods are $(\varepsilon=5,\delta=10^{-6})$-DP. For width, shaded bands indicate $\pm 2$ empirical standard errors (SE) across repetitions; for coverage, shaded bands indicate $\pm 2$ binomial SE. Details are given in Section~\ref{sec:num_study}. }
\label{fig:main_fig}
\end{figure*}

In parametric settings, where we have an analytically characterized (or asymptotically approximated) sampling distribution for the statistic, CIs are typically constructed by estimating the distribution's parameter. For example, a mean statistic follows a normal limiting distribution, and its variance can be estimated directly from the data. However, these assumptions do not hold for many practical statistical tasks. In nonparametric settings, or when estimating the limiting parameters is challenging, resampling methods such as the bootstrap \cite{efron1992bootstrap} provide a general route to uncertainty quantification. Given a dataset, the bootstrap method draws multiple samples from it, each of the original sample size, by sampling with replacement, and recomputes the statistic on each. Intuitively, this procedure attempts to mimic the variability one would see if one could resample the initial dataset. Although resamples are not independent new datasets, under mild conditions the empirical distribution of the resampled statistics consistently approximates the sampling distribution of the original statistic. Consequently, taking its $\alpha/2$ and $1-\alpha/2$ empirical quantiles yields an approximate $(1-\alpha)$-CI.

\textbf{Differential privacy} When data is sensitive, such as medical records, statistical procedures must account not only for inferential accuracy but also for the privacy of the individuals represented in the data. Differential privacy (DP) \citep{dwork2006calibrating}, offers a rigorous framework that has become the gold standard both in academia and in industry for privacy-preserving data analysis. DP provides strong, mathematically provable privacy guarantees quantified by two parameters $\varepsilon > 0$ and $\delta \in [0,1]$,  requiring that w.p. at least~$\approx 1- \delta$ the probability of any class of outputs can increase by a factor of at most $e^{\varepsilon}$ when changing a single element in the dataset.

In parametric settings, CIs can often be derived by privately estimating the relevant distributional parameters~\citep{du2020differentially, karwa2018finite, wang2018statistical, shoham2022asking}. In nonparametric settings, or when the limiting distribution is difficult to characterize or estimate, CIs can sometimes be constructed for specific functionals by reformulating the interval endpoints as estimable distributional quantities. For example, \citet{drechsler2022nonparametric} provides nonparametric DP CIs for the median by directly estimating other quantiles. Both cases require analysis of the specific statistic and its distribution, and do not fall in our general framework.

In the classical bootstrap, each element might participate in the resampled dataset more than once, increasing the query's sensitivity to changing a single element, which worsens the privacy-accuracy tradeoff. \citet{brawner2018bootstrap} address this by capping the number of times any element may appear in a bootstrap sample. Since the cap is rarely triggered, the procedure preserves classical bootstrap behavior with high probability. Nevertheless, the cap induces a constant-factor increase in sensitivity, which can noticeably degrade accuracy in finite samples. \citet{wang2025differentially} further refine the privacy analysis by characterizing the privacy loss of a single bootstrap sample and its composition using the $f$-DP framework. However, their validity guarantees are restricted to the Gaussian mechanism, and in practice, their numerical privacy analysis method can account only for Gaussian DP. While the relative degradation in CI width caused by bootstrap privatization decreases as the number of bootstrap samples tends to infinity, we show that the degradation is significantly worse for reasonable sample sizes. Finally, \citet{dette2025gaussian} study the $m$-out-of-$n$ bootstrap, in which subsets of smaller size $m<n$ are sampled \emph{with} replacement. Their analysis is restricted to the Gaussian mechanism and to statistics with sensitivity $O(1/n)$. Coverage guarantees are derived under the assumption of a continuous underlying distribution. The privacy analysis for finite samples requires numerically evaluating convolutions, while closed-form bounds are provided asymptotically.

In recent years, a technique known as Bag-of-Little-Bootstrap (\textsc{BLB}), proposed by \citet{kleiner2014scalable}, has emerged as a valuable tool for non-private bootstrapping of large databases. It splits the data into disjoint subsets, bootstraps samples of the original size within each subset, and aggregates the resulting CDF or standard-deviation estimates. BLB subset estimates can also be aggregated \emph{privately} using the subsample-and-aggregate technique \citep{nissim2007smooth}, for example, using the CoinPress mechanism \citep{covington2025unbiased} or private median estimation \citep{chadha2024resampling}. Elegantly, the subsample-and-aggregate technique does not generally require a private estimator of the target quantity, but only of the aggregation. 
However, both \cite{covington2025unbiased} and \cite{chadha2024resampling} require asymptotic normality of the statistic to establish validity. More importantly, since their methods use sample splitting, each split must be sufficiently representative of the underlying distribution, which is hard to achieve with small datasets when the underlying distribution is more challenging, as illustrated in Figure \ref{fig:main_fig}.

\begin{wrapfigure}{R}{0.6\textwidth}
\begin{minipage}[t]{0.6\textwidth}
\vspace{-1cm}
\begin{algorithm}[H]
\caption{\privSub: Private Subsampling for Quantile CI}
\label{alg:privsub}
\begin{algorithmic}
\Require Dataset $\dataset=(\omega_1,\ldots,\omega_n)$, subsample size $m\in[n]$, number of subsamples $T\in \left[\binom{n}{m}\right]$, rate $\tau(\cdot)$, level $\alpha\in(0,1)$, DP mechanisms $\mathcal{M}_{\text{full}}$ and $\mathcal{M}_{\text{sub}}$.
\State $\widetilde{\theta} \gets \mathcal{M}_{\text{full}}(\dataset)$
\For{$i=1$ to $T$}
    \State Draw subsample $S_i \subseteq \dataset$ with $|S_i|=m$
    \State $\widetilde{\theta}_i \gets \mathcal{M}_{\text{sub}}(S_i)$
\EndFor
\State Sort $\{\widetilde{\theta}_i\}_{i=1}^T$ to obtain $\widetilde{\theta}_{(1)} \le \cdots \le \widetilde{\theta}_{(T)}$
\State $q_{\alpha/2} \gets \max \left(1,\lfloor (\alpha/2)T \rfloor \right)$, \quad $q_{1-\alpha/2} \gets \lceil (1-\alpha/2)T \rceil$
\State \Return $\Big[\widetilde{\theta}-\frac{\tau_m}{\tau_n}\big(\widetilde{\theta}_{(q_{1-\alpha/2})}-\widetilde{\theta}\big),\ \widetilde{\theta}+\frac{\tau_m}{\tau_n}\big(\widetilde{\theta}-\widetilde{\theta}_{(q_{\alpha/2})}\big)\Big]$
\end{algorithmic}
\end{algorithm}
\vspace{-.5cm}
\end{minipage}
\end{wrapfigure}

\textbf{Our approach} \ \ We overcome the data-reuse challenge by leveraging a different sampling method, known as subsampling \citep{politis1999subsampling}.
Given a dataset of size $n$, this method, parameterized by $m \in [n]$ and $T \in \left[\binom{n}{m}\right]$, first computes the statistic on the full dataset.
Then, it samples $T$ subsets of size $m$ without replacement and uses them to obtain $T$ statistics. This results in an estimation of the CDF of the statistic on $m$ samples.
Centering around the full dataset statistic and rescaling by a normalizing factor produces an estimation of the CDF of the statistic on $n$ samples. The CI can be constructed by taking the $\alpha/2$ and $1-\alpha/2$ quantiles of the empirical CDF, or by estimating the standard deviation of the statistic.

Surprisingly, this method results in valid CIs under milder assumptions than bootstrapping \citep{bickel2012resampling}. It relies on the fact that the sequence of distributions of statistic estimations on $n$ samples multiplied by an appropriate factor $\tau_{n}$, converges to some target distribution. We refer to the sequence of factors $\tau_{n}$ as the \emph{convergence rate} and to the target distribution as the \emph{limiting distribution}. For example, the mean converges to a normal distribution at convergence rate $\tau_n=\sqrt{n}$. For sufficiently large $m$ and $T$, the $T$ statistics accurately approximate the underlying distribution of the statistic on $m$ elements. Rescaling by $\tau_{m}/\tau_{n}$ results in an accurate estimation of the CDF of the statistic on $n$ samples.  Our private CI approach, presented in Algorithm~\ref{alg:privsub}, makes the subsampling algorithm private by replacing each statistic evaluation with a private estimate.

\textbf{Contributions}
We propose a privacy-preserving mechanism for constructing nonparametric CIs for any quantity of interest under minimal distributional assumptions, and prove it is differentially private (Theorem \ref{thm:eps_delta_total}) and asymptotically accurate (Corollary \ref{cor:privmed_valid_tight}). In fact, our method produces an asymptotically accurate estimate of the full CDF of the statistic (Theorem \ref{thm:cons_priv_quan_CI}), which, by itself, has many applications, such as estimating multiple moments and testing assumptions. We demonstrate that our algorithm performs well at practical sample sizes (Figure~\ref{fig:main_fig} and Appendix~\ref{subsec:full_emp_eval}).

Our validity guarantees hold under the standard setting of the non-private subsampling literature. We require the \textit{existence} of a limiting distribution and an arbitrary but known convergence rate.\footnote{Knowledge of the convergence rate is used only to obtain tight, rate-corrected CIs. Valid CIs can still be constructed without this knowledge, with more conservative scaling. Moreover, this assumption can be relaxed by estimating the convergence rate, as in \citet{bertail1999subsampling}; this extension is outside the scope of the present paper.} In addition, we require that the error introduced by the privacy-preserving mechanism is asymptotically lower order than the statistical error due to sampling (Definition~\ref{def:tau_n_cons}). Our method provides both pure ($\delta = 0$) and approximate ($\delta \in [0,1]$) DP, though its main asymptotic advantage comes from using advanced composition, which yields approximate DP.

These assumptions are weaker than those used by existing general, nonparametric private CI methods. In particular, prior works assume either a normal limiting distribution with a known convergence rate~\citep{chadha2024resampling, wang2025differentially}, or the specific convergence rate $\tau_n=\sqrt{n}$~\citep{dette2025gaussian}. The lower-order perturbation requirement is also standard in this line of work~\citep{chadha2024resampling, dette2025gaussian, wang2025differentially}. It is satisfied by many standard privacy-preserving mechanisms under typical choices of the privacy budget, although, naturally, whether it holds depends on the privacy level.

We allow any choice of $m$ and $T$ satisfying
$m/n \rightarrow 0$, $m \rightarrow \infty$, $\tau_m/\tau_n \rightarrow 0$, and $T\rightarrow \infty$ as $n\rightarrow \infty$. Because no general finite-sample rates are available under these minimal distributional assumptions, even in the non-private setting, we cannot derive universally optimal choices of $m$ and $T$. Instead, we provide in Section~\ref{sec:priv_and_acc_privsub} several heuristic choices, and discuss them in more detail in Appendix~\ref{sec:rates}.

\vspace{-0.25cm}
\paragraph{Empirical evaluation} In Figure \ref{fig:main_fig}, we compare several methods for constructing CIs for the median, logistic regression coefficient (both with normal limiting distribution), and the Kolmogorov-Smirnov (KS) statistic (non-normal limiting distribution) with a target confidence level of $(1-\alpha)=0.9$. We compare our method with the other two general, nonparametric DP methods: $\blb$ \citep{chadha2024resampling}, and \privBoot\ \citep{wang2025differentially}. The three methods privately estimate the median using the inverse sensitivity mechanism \citep{asi2020instance}, the logistic regression coefficient using output perturbation with Gaussian noise addition \cite{wang2019differentially}, and the KS statistic with noise addition (Laplace \textendash{} for \privSub\ and \blb, and Gaussian \textendash{} for \privBoot).
 For median estimation, we also include a private baseline, $\expMech$ \citep{drechsler2022nonparametric}.\footnote{We use output perturbation rather than objective perturbation because \privBoot\ is formulated under $\mu$-GDP, and we wanted to compare all three methods under a fair privacy framework. Our goal in this experiment is not to use the state-of-the-art private estimator for logistic regression, but rather to provide a fair comparison between the resulting CI procedures.}

Figure~\ref{fig:main_fig} illustrates the main empirical message of our work. Across the settings we consider, \privSub\ produces valid CIs whose widths approach those of the corresponding non-private baseline. Meanwhile, existing private resampling-based methods either undercover or become overly conservative at finite sample sizes, and are not robust under different statistics and distributions. Overall, our experiments demonstrate the flexibility of our framework beyond asymptotically normal settings, alongside its strong finite-sample performance. We provide a detailed discussion of these experiments in Section~\ref{sec:num_study} and additional empirical results in Appendix~\ref{app:detailed_analysis}.

\section{Preliminaries and notation}\label{sec:preliminaries}

Let $(\Omega,\mathcal{F},P)$ be a probability space and $P^{(n)}$ be the product distribution. We use \textbf{bold} symbols to denote vectors, e.g., $\dataset =(\omega_1,\ldots,\omega_n) \sim P^{(n)}$. Unless specified otherwise, $\prob$ denotes the probability taken with respect to the joint distribution over all sources of randomness, and limits are taken as $n\to\infty$. We aim to infer a distributional functional such as a mean, quantile, or correlation. To estimate this quantity, we define a function, $\theta:\Omega^* \rightarrow \R$; which, together with $P$, defines a random variable. We will denote the quantity of interest by $\theta^*$, formally defined by $\theta^* = \underset{n\rightarrow \infty}{\lim} \underset{\dataset \sim P^{(n)}}{\E}\theta(\dataset)$.
Given some constant $\tau \in \R^+$, we define the standardized centered cumulative distribution function of the random variable $\theta(\dataset)$ at the point $x \in \R$ as
\begin{equation}\label{eq:U_n(x,P)}
    U_{n, \tau}(x) = \underset{\dataset \sim P^{(n)}}{\prob}\left( \tau \cdot \bigl(\theta(\dataset) - \theta^* \bigr) \leq x \,\right).
\end{equation}
It is common practice to assume that there exists a non-decreasing sequence $\tau_{n} \rightarrow \infty$ and a distribution $U(\cdot)$, such that $U_{n, \tau_n}(\cdot) \to U(\cdot)$ pointwise. We refer to $\tau_n$ as the convergence rate and to $U(\cdot)$ as the limiting distribution, and denote $U_{n}(\cdot) \coloneqq U_{n, \tau_n}(\cdot)$. For example, under some conditions, means (by CLT), maximum likelihood and M-estimators, quantiles (including the median), U-statistics, smooth functionals, and empirical risk minimizers, all have a normal limiting distribution at rate $\sqrt{n}$ \cite{van2000asymptotic}.

To formally define what \emph{approaches} means here, we must define a notion of convergence of distributions. For random variables $X_n,X$ having some joint distribution and taking values in a metric space, we write $X_n \pto X$ (``$X_n$ converges to $X$ in probability'') if, for every $\xi>0$,
$\lim_{n\to\infty}\mathbb{P}(|X_n-X|>\xi)=0.$ Letting
$F_{X_n}(x)=\mathbb{P}(X_n\le x)$ and $F_X(x)=\mathbb{P}(X\le x)$, we write $X_n \xrightarrow{d} X$ (``$X_n$ converges to $X$ in distribution'') if
$\underset{n\to\infty}{\lim} F_{X_n}(x) = F_X(x)$
 for all continuity points of $F_X$. It can be shown that convergence in probability implies convergence in distribution.

With these notations and definitions, we can now formally define confidence intervals for the typical one-dimensional setting we study.

\begin{definition}[Asymptotically valid and tight confidence intervals]\label{def:con_int}
    Let $P$ be some distribution defined on $\Omega$ and $\theta : 
    \Omega^{*} \to \R$ be some function. Given a CI construction method from a sample of size $n$, denote by $(u_{n}(\dataset), v_{n}(\dataset))$ the random variables defining its edges. The coverage of $(u_{n},v_{n})$ is defined as $C_{n} \coloneqq \prob_{\dataset \sim P^{(n)}}\left(\theta^*  \in [u_{n}(\dataset), v_{n}(\dataset)]\right)$.     
    We say this method is an asymptotically valid $(1-\alpha)$-CI of $\theta$ if $\lim_{n\rightarrow \infty} C_{n} \ge 1-\alpha$ and asymptotically tight if $\lim_{n\rightarrow \infty} C_{n} \le 1-\alpha$.
\end{definition}

Validity and tightness capture fundamentally different aspects of a confidence interval and are not equivalent. Asymptotic validity is a one-sided, feasibility requirement ensuring that the nominal coverage level is not violated in the limit. In contrast, tightness is a calibration property, quantifying how closely the procedure attains the target level. In general, we can use any empirical distribution to construct valid and tight CIs, provided that this empirical distribution converges, in a sufficiently strong sense, to the limiting distribution of the statistic estimating the quantity of interest. Lemma~\ref{lem:val_CI_emp_dist} gives the formal statement.

\subsection{Subsampling}\label{sec:non_priv_subsampling}
Consider any function $m: \mathbb{N} \rightarrow \mathbb{N}$ such that $m(n) \in [n]$, $m(n) \rightarrow \infty$, $ m(n)/n \to 0$, and $\tau_m / \tau_n \to 0$ as $ n \to \infty$.\footnote{This last condition is always achievable since $\tau_n \rightarrow \infty$.} We refer to the setting where a limiting distribution exists and $m$ is chosen according to these conditions as the \emph{standard subsampling setting}. We denote $m \coloneqq m(n)$ for brevity.

Let $\tau_n$ be the convergence rate such that $U_n$ has a limiting distribution $U$, and define $\tau_m$ accordingly.  Given a subset size $m$, we denote by $I$ a subset of indices, that is, $I=(I_1,...,I_m)\subseteq [n]$, $|I|=m$. Given a sample $\dataset$, we will denote by $\dataset(I)=(\omega_{I_1},\ldots,\omega_{I_m})$, a subset of the data, indexed by $I$. A list of subset indices will be denoted by $\boldsymbol{I}$.

Given a dataset $\dataset$, list of subset indices $\boldsymbol{I}$, and a constant $\tau \in \R^+$, we define the empirical cumulative probability distribution at a point $x\in \R$ by 
\begin{equation}\label{eq:U_sub}
U_{\tau}(x; \dataset, \boldsymbol{I})  \coloneqq \frac{1}{|\boldsymbol{I}|}\sum_{I\in \boldsymbol{I}}\mathbbm{1}\{\tau \cdot (\theta(\dataset(I))-\theta(\dataset)) \leq x \}.
\end{equation}
We denote by $\widehat{U}_{n,m}(x)$ the random variable that draws $\dataset \sim P^{(n)}$ and then plugs it into Equation \ref{eq:U_sub}, with $\tau=\tau_m$ and $\boldsymbol{I} = (I_1,\ldots, I_{\binom{n}{m}})$ covering all possible choices of subsets. $\widehat{U}_{n,m}(\cdot)$ is the point-wise estimate of $U_{n}(\cdot)$ simultaneously over all $x \in \R$ \eqref{eq:U_n(x,P)}, estimated using all possible subsamples of size $m$.

We now have all the notation and definitions required to state the main theorem of consistency of non-private subsampling under minimal distributional assumptions.
\begin{theorem}[Adapted from Theorem 2.2.1 in \cite{politis1999subsampling}]\label{thm:non_priv_ci_consis}
Under the standard subsampling setting (\ref{sec:non_priv_subsampling}), we have that $\widehat{U}_{n,m}(x) \pto U(x)$ for any continuity point $x$ of $U(x)$. Furthermore, if $ U(\cdot)$ is  continuous, then $\sup_x |\widehat{U}_{n,m}(x) - U(x)| \pto 0.$
\end{theorem}

Theorem \ref{thm:non_priv_ci_consis} ensures uniform convergence in probability (or point-wise if the distribution of $U(\cdot)$ is not continuous) as long as the subsample size $m$ goes to infinity slower than $n$. By Lemma \ref{lem:val_CI_emp_dist} this implies that if $U(\cdot)$ is continuous at the $\alpha/2$ and $1-\alpha/2$ quantiles, then estimating the quantiles from $\widehat{U}_{n,m}(\cdot)$ will give an asymptotically valid and tight $(1-\alpha)$-CI. The proof of Theorem \ref{thm:non_priv_ci_consis} is based on two observations. The first is that $\theta(\dataset)$ converges to $\theta^*$ faster than $\theta(\dataset(I))
 $, which means that we can replace $\theta(\dataset)$ in \eqref{eq:U_sub} by $\theta^*$ (the error is negligible asymptotically). 

Summing over all $\binom{n}{m}$ subsets is computationally burdensome, but one can use a stochastic approximation instead, and use the Dvoretzky, Kiefer, Wolfowitz inequality (see \cite{serfling2009approximation}). Define the random variable $\widehat{U}^T_{n,m}(x)$ given by sampling $\dataset \sim P^{(n)}$, sampling with replacement $T$ subsets $(\boldsymbol{I}=I_{(1)},...,I_{(T)})$ of size $m$, where each subset consists of $m$ elements drawn without replacement, and plugging both into Equation \eqref{eq:U_sub}.

\begin{theorem}[Corollary 2.4.1 in \cite{politis1999subsampling}]\label{thm:stoch_nonpriv_subsample_valid}
    The results of Theorem \ref{thm:non_priv_ci_consis} hold when replacing $\widehat{U}_{n,m}(x)$ by $\widehat{U}^{T}_{n,m}(x)$, as long as $T\rightarrow \infty$.
\end{theorem}

\subsection{Differential privacy}

We give a very brief introduction to differential privacy here. A more detailed description can be found in Appendix \ref{app:DP}. Given some domain $\Omega$ and a sample size $n \in \N$, we call two datasets $\dataset, \dataset' \in \Omega^n$  neighbors, denoted by $\dataset \sim \dataset'$, if they are identical except for one of their elements.

\begin{definition}[Differential privacy]
        Given $\varepsilon \ge 0$, $\delta \in [0,1]$, a data domain $\Omega$ and some domain of responses $\mathcal{R}$, we say that a mechanism ${\cal M}:\Omega^n \rightarrow \mathcal{R}$ satisfies $(\varepsilon, \delta)$-\textit{differential privacy}, denoted by  $(\varepsilon, \delta)$-DP, if 
        $\prob({\cal M}(\dataset) \in E)\le e^\varepsilon \prob({\cal M}(\dataset') \in E)+\delta $
        for all $\dataset \sim \dataset' \in \Omega^n$ and all $E \subseteq \mathcal{R}$.
        When $\delta = 0$ we say $\cal M$ is \emph{pure} DP, and if $\delta > 0$ $\cal M$ is said to be \emph{approximate} DP.
\end{definition} 

Differential privacy enjoys several useful properties. First, it holds under post-processing; that is, if an algorithm is differentially private, then any follow-up analysis of the algorithm's output without additional access to the dataset cannot degrade the privacy guarantee (Proposition 2.1 in \cite{dwork2014algorithmic}). DP also composes well; that is, if we consecutively apply multiple differentially private mechanisms to the same dataset, the overall privacy loss can be bounded with linear and sublinear increase in privacy parameters for pure and approximate DP, respectively (Lemmas \ref{lem:basic_com}, \ref{lem:adv_com}). A third useful fact is that privacy is amplified by subsampling; that is, if a differentially private mechanism is only applied to a random subset of the dataset, the privacy guarantees of the mechanism are amplified by approximately the chance of each element to appear in the subset (Lemma \ref{lem:amp_by_sub}).

Privately estimating the mean and other moments can be done using noise addition mechanisms such as Laplace and Gaussian (Definitions \ref{Def:Lap_mec}, \ref{def:Gaus_mec}), if the sensitivity of the statistic with respect to a change of a single element can be bounded. The median and other quantiles can be privately estimated using the inverse sensitivity mechanism (\ref{def:inv_sen}), an instantiation of the general exponential mechanism (\ref{def:exp_mech}) which can be used to estimate any quantity if its utility's sensitivity is bounded.

\section{Privacy and accuracy of \privSub}\label{sec:priv_and_acc_privsub}

In this section, we analyze the privacy and accuracy (validity and tightness) of our proposed algorithm. 

\textbf{Privacy} We start by providing privacy guarantees for $\privSub$.
\begin{theorem}\label{thm:eps_delta_total}
    If
    $\mathcal{M}_{\operatorname{full}}$ is $(\varepsilon_{\operatorname{full}},\delta_{\operatorname{full}})$-DP, and $\mathcal{M}_{\operatorname{sub}}$ is $(\varepsilon_{\operatorname{sub}},\delta_{\operatorname{sub}})$-DP, for some $\varepsilon_{\operatorname{full}},\varepsilon_{\operatorname{sub}} \ge 0$; $\delta_{\operatorname{full}}, \delta_{\operatorname{sub}} \in [0,1]$, denoting
   $ \varepsilon_{\operatorname{amp}} = \log\left(1 + \frac{m}{n} \left(e^{\varepsilon_{\operatorname{sub}}} - 1\right)\right)$, $\delta_{\operatorname{amp}}= \frac{m}{n} \delta_{\operatorname{sub}} $,
   $\privSub$ is $(T \varepsilon_{\operatorname{amp}} + \varepsilon_{\operatorname{full}}, T \delta_{\operatorname{amp}} + \delta_{\operatorname{full}})$-DP; and for some $\delta'' > 0$ it is $(\varepsilon, T\delta_{\operatorname{amp}} + \delta'' + \delta_{\operatorname{full}})$-DP, where
   $$ \varepsilon = \varepsilon_{\operatorname{amp}} \cdot \left( \sqrt{2T \log(1/\delta'')}  + T \left(\frac{e^{\varepsilon_{\operatorname{amp}}} - 1}{e^{\varepsilon_{\operatorname{amp}}} + 1} \right) \right)  + \varepsilon_{\operatorname{full}}.$$
\end{theorem}
    
\begin{proof}
    By the basic composition lemma (\ref{lem:basic_com}) and the fact that $\mathcal{M}_{\operatorname{full}}$ (line 1) is $(\varepsilon_{\operatorname{full}}, \delta_{\operatorname{full}})$-DP, it suffices to bound the privacy loss resulting from the repeated calls to $\mathcal{M}_{\operatorname{sub}}$ (line 4). Amplification by subsampling (Lemma \ref{lem:amp_by_sub}) implies each of these calls is $(\varepsilon_{\text{amp}}, \delta_{\text{amp}})$-DP. Combining this with basic or advanced composition (Lemmas \ref{lem:basic_com}, \ref{lem:adv_com}) completes the proof.
\end{proof}
\vspace{-0.2cm}
We note that when $\varepsilon_{\text{sub}} \le 1$ we have $\varepsilon_{\text{amp}} \approx \frac{m}{n}\varepsilon_{\text{sub}}$, which implies $\varepsilon \approx \frac{m T}{n}\varepsilon_{\text{sub}} + \varepsilon_{\text{full}}$ using basic composition, and $\varepsilon \approx \frac{m}{n} \sqrt{2T \cdot \ln(1/\delta'')} \varepsilon_{\text{sub}} + T \left(\frac{m}{n} \right)^2\varepsilon_{\text{sub}}^{2} + \varepsilon_{\text{full}}$ using advanced composition. We focus on regimes with $mT \gtrsim n$, so that the total number of sampled records across all subsamples is at least on the order of the original sample size. Thus, using basic composition implies $\varepsilon_{\text{sub}}$ decreases as $n$ grows, but as long as $m\cdot\sqrt{T}<n$, it increases when using advanced composition. Though decreasing $\varepsilon_{\text{sub}}$ implies an increase in the scale of the perturbation of the mechanism, the increase in sample size mitigates this effect in many parameter regimes.

\textbf{Validity and tightness} We now turn to prove that the private CI based on the quantile method is asymptotically valid and tight, as defined in Definition \ref{def:con_int}. Fixing a dataset $\dataset$ and a sequence of indices subsets $\boldsymbol{I}$, we denote
\begin{equation}\label{eq:cen_stat_sub_CDF_priv_T}
\widetilde{U}_{\tau}(x; \dataset, \boldsymbol{I}) \coloneqq 
\frac{1}{|\boldsymbol{I}|}\sum_{I \in \boldsymbol{I}}
\mathbbm{1}\left\{
\tau \cdot
\left(
{\cal M}_{\mathrm{sub}}(\dataset(I))
-
{\cal M}_{\mathrm{full}}(\dataset)
\right)
\le x
\right\}.
\end{equation}
Unlike Equation~ \eqref{eq:U_sub}, Equation \eqref{eq:cen_stat_sub_CDF_priv_T} is a random quantity, due to $\cal M$'s internal randomness. We use this equation to define the private counterpart of $\widehat{U}_{n,m}(\cdot)$, denoting by $\widetilde{U}_{n,m}^{T}(x)$ the random variable that draws $\dataset \sim P^{(n)}$, samples $T$ subsets $(\boldsymbol{I}=I_{(1)},...,I_{(T)})$ of size $m$ randomly with replacement, then plugs it into Equation \ref{eq:cen_stat_sub_CDF_priv_T}, with $\tau=\tau_m$.

To derive the private equivalent of Theorem \ref{thm:non_priv_ci_consis}, it is essential to have accuracy guarantees for the private estimators as formalized in Definition \ref{def:tau_n_cons}. Intuitively, we require that the privacy-induced perturbation is lower order compared to the convergence rate $\tau_n$; that is, the typical distance between the private and non-private estimators decreases faster than the typical distance between the non-private estimator and the quantity of interest. An analysis of the regime where the perturbation error is larger than the statistical error would require specific accuracy  (e.g., CI width) and knowledge of the exact perturbation mechanism and the statistic, and therefore cannot be analyzed under this general framework. Thus, a similar assumption is made in \cite{chadha2024resampling} (rate-$\sqrt{n}$-resampling consistent, Equations 6 and 7),  \cite{wang2025differentially} and \cite{dette2025gaussian} explicitly consider only Gaussian noise addition with bounded sensitivity, which satisfy this property.

\begin{definition}[$\tau_n$-consistency]\label{def:tau_n_cons} 
    Given a convergence rate $\tau_n\to\infty$, we say a mechanism ${\cal M}$ is \emph{$\tau_n$-consistent} if $\tau_n \cdot (\theta(\dataset)-{\cal M}(\dataset))$ converges in probability to $0$. 
\end{definition}

This property is commonly achieved in various settings. The statistical error of many quantities of interest, such as moments, quantiles, and regression coefficients, scales as $O(1/\sqrt{n})$, while the privacy-induced error typically scales like $O(1/(n\varepsilon))$ \cite{asi2020instance, chadha2024resampling, wang2019differentially}. This ensures that, as the sample size grows, the contribution of the statistical signal dominates the injected noise. We give two examples of private mechanisms that are $\tau_n$-consistent: Noise-addition mechanisms used for mean and other moments estimation (Claim \ref{clm:lap_vanish}) and the inverse sensitivity mechanism \cite{asi2020instance} used for median and other quantile estimations (Claim \ref{clm:expmech_median_fixed}). From this point on, we assume that the private estimators, together with their privacy parameters, satisfy the consistency requirement, and we proceed to analyze the CIs produced by $\privSub$.

\begin{theorem}\label{thm:cons_priv_quan_CI}
The results of Theorem \ref{thm:non_priv_ci_consis} hold when replacing $\widehat{U}_{n,m}(x)$ by $\widetilde U^{T}_{n,m}(x)$ as long as the ${\cal M}_{full}$ and ${\cal M}_{sub}$ are $\tau_n$ and $\tau_m$-consistent (Definition \ref{def:tau_n_cons}), respectively, and $T \rightarrow \infty$.
\end{theorem}

Theorem \ref{thm:cons_priv_quan_CI} immediately implies the following Corollary by invoking Lemma \ref{lem:val_CI_emp_dist}

\begin{corollary}\label{cor:privmed_valid_tight}
    If the limiting distribution $U$ is continuous at its $\alpha/2$ and $1-\alpha/2$ quantiles, then the CI returned by \privSub\ is asymptotically valid and tight.
\end{corollary}

Unlike $U_{\tau}(x; \dataset, \boldsymbol{I})$ (Equation~\eqref{eq:U_sub}), which is a U-statistic, $\widetilde{U}_{\tau}(x; \dataset, \boldsymbol{I})$ (Equation \eqref{eq:cen_stat_sub_CDF_priv_T}) is a random quantity, so the proof technique of Theorem \ref{thm:non_priv_ci_consis} does not apply. But, since this randomness is independent of the data, we can decompose the centered private estimator into the non-private centered statistic that converges to the limiting distribution (Theorem \ref{thm:non_priv_ci_consis}), and additional perturbation terms that tend to $0$ in probability. The complete proof of Theorem \ref{thm:cons_priv_quan_CI} can be found in Appendix~\ref{app:miss_proof}. 

\paragraph*{On choosing $m$ and $T$}
As discussed above, under the minimal distributional assumptions, it is impossible to derive general finite-sample rates, even in the non-private setting, and therefore impossible to optimize the hyperparameters $m$ and $T$ in full generality. Nevertheless, to provide some guidance, we can consider a simpler regime in which the statistic admits a limiting normal approximation, as assumed by other nonparametric methods \cite{chadha2024resampling, wang2025differentially, dette2025gaussian}. Specifically, we consider a variance-based CI that extends the appropriate normal quantile times the estimated standard deviation (SD) to the left and right of the private point estimate. In this setting, one can derive a simple heuristic rate for \emph{estimating the SD} under an additive noise-addition mechanism, as in \cite{wang2025differentially,dette2025gaussian}. The derivation is given in Appendix~\ref{sec:rates}. 

If the goal is to minimize the \emph{expected width} of a normal-approximation CI, then increasing either the number of subsamples $T$ or the subsample size $m$ does not reduce the target variance itself. Rather, larger values of $T$ or $m$ mainly increase the privacy cost, and hence the magnitude of the privacy perturbation (see Figures~\ref{fig:varying_m} and \ref{fig:varying_T}). This suggests the natural constraint $m \cdot T \gtrsim n$, which ensures that the subsampling procedure uses the data at roughly the scale of the original sample without incurring unnecessary privacy cost. Under this constraint, the privacy-induced contribution is minimized by taking $m \propto \frac{n}{T}$.
This scaling is consistent with the choice suggested by \citet{dette2025gaussian}. If, instead, the objective is to minimize the \emph{mean squared error} of the variance estimator itself, then a different hyperparameter regime emerges. The heuristic analysis then favors taking $m=\Theta(n)$, and optimizing over the number of subsamples yields $T \propto \varepsilon^2 n$, matching the scaling obtained in \cite{wang2025differentially}.

In our experiments, we take $T=60$, which is large enough to obtain stable empirical quantiles, (see Figures~\ref{fig:varying_significance_width} and \ref{fig:varying_significance_coverage}),
and $m=n^{2/3}$, as suggested by \cite{politis1999subsampling}, although our numerical results are robust to a wide range of hyperparameters (see Appendix \ref{subsec:full_emp_eval}). In our parameter regime, this choice always satisfies $T \cdot m > n$. More generally, the choice of $m$ reflects a tradeoff: increasing $m$ improves the statistical fidelity of the subsample estimates to the full-sample estimator, but weakens privacy amplification by subsampling, thereby increasing the privacy-induced noise required for each estimate. Consequently, the optimal choice of $m$ can vary substantially depending on whether privacy noise is a dominant source of error or is negligible relative to the statistical error. When the privacy perturbation is significant, or when the statistic is highly non-regular or nonlinear, we suggest taking smaller subsamples, with $m$ closer to $\sqrt{n}$, to benefit more from amplification and to better capture nonstandard behavior. Conversely, when privacy noise is relatively small or the underlying distribution is challenging and larger subsamples are needed for stable estimation, one may benefit from taking $m$ larger.

\section{Empirical evaluation}\label{sec:num_study}

In this section, we summarize the numerical evaluation shown in Figure~\ref{fig:main_fig}. We compare \privSub\ with two existing general nonparametric DP CI methods: the BLB-based approach \blb\ \citep{chadha2024resampling} and the private bootstrap \privBoot\ \citep{wang2025differentially}. For each setting, sample size $n$, and method, we generate $1000$ independent datasets and construct nominal $90\%$ CIs. We report the empirical coverage and the average CI width, measuring validity and finite-sample tightness, respectively. We set $T=60$ for both \privSub\ and \privBoot, and use $m=n^{2/3}$ for \privSub. Other method-specific hyperparameters are chosen following the corresponding papers. When a method requires splitting the privacy budget across subroutines, we use an equal split.\footnote{Optimizing this split can improve finite-sample constants; Appendix~\ref{subsec:full_emp_eval} reports a sensitivity analysis. To avoid per-setting tuning and keep the comparison simple, we use an equal split throughout.} All methods use a privacy budget of $\varepsilon=5$ for all experiments. Since \privBoot\ cannot guarantee pure DP, we set $\delta=10^{-6}$, while other methods guarantee pure DP for median and KS-distance. Additional implementation details and experiments are given in Appendix~\ref{app:detailed_analysis}.

\textbf{Median estimation.}
We first consider CI construction for the median under a truncated normal distribution, a continuous, unimodal distribution, a bit asymmetric due to the truncation, the setting considered in \cite{chadha2024resampling}. In addition to the general methods \privSub, \blb, and \privBoot, we include \expMech\ \citep{drechsler2022nonparametric}, a median-specific private baseline based on privately estimating appropriate quantiles. For the resampling-based methods, the underlying DP median estimate is obtained using the inverse sensitivity mechanism.

\textbf{Logistic regression}.
We next study inference for the slope parameter in a simple logistic regression model with one covariate and an intercept, following \citet{wang2025differentially}. For each observation, we sample $\omega_i \sim \operatorname{Unif}[0,1]$, set $x_i=(1,\omega_i)^\top$, and generate a binary response from the logistic model with data-generating parameter $\theta=(0,0.8)$. The target parameter is the slope coordinate of the population minimizer of the corresponding regularized logistic risk. DP estimates are obtained using output perturbation for regularized logistic regression \citep{wang2019differentially}.

\textbf{The Kolmogorov--Smirnov statistic}.
Finally, we evaluate CI construction for the KS statistic
$
D_n=\sup_x |F_n(x)-F(x)|.
$
We focus on uniformity testing where $F$ is the uniform distribution, where the data are sampled from $F$. This example is meant to test performance for a non-smooth statistic, with a non-normal limiting distribution; therefore, the non-private baseline is subsampling. The KS statistic has global sensitivity $1/n$ \cite{awan2025differentially}, so it can be privatized directly using standard noise-addition mechanisms. 

The runtime of our algorithm matches that of the corresponding non-private resampling baseline, up to the additional cost of the privacy-preserving mechanism. Existing private nonparametric methods \cite{wang2025differentially,dette2025gaussian,chadha2024resampling} exhibit the same type of dependence. The complexity also scales with $T$, the number of subsets, but we propose choosing a small $T$, $60$ in our experiment. In practice, for all parameter regimes we considered, each task ran in at most a few seconds on a standard PC.

{\bf Results summary.}
Across all three experiments, \privSub\ maintains near-nominal coverage, is sometimes mildly conservative, and CI width approaches that of the corresponding non-private baseline as $n$ grows, while \privBoot\ is consistently conservative, resulting in a wider CI, and \blb's coverage is less stable across settings. In the \textbf{left panel} of Figure~\ref{fig:main_fig}, \privSub\ achieves valid coverage with widths close to the non-private baseline. The median-specific private baseline is also competitive, as expected, but is tailored to this particular functional. By contrast, \blb\ under-covers at the sample sizes considered, leading to invalid CIs, while \privBoot\ is substantially more conservative and produces wider intervals. In the \textbf{center panel}, \blb\ performs well and nearly matches the non-private baseline under the favorable hyperparameter choice used in this experiment. However, unlike \privSub, its performance is less robust across the non-smooth and more challenging settings considered here. The \textbf{right panel}, highlights the main advantage of our approach: the KS statistic is a non-smooth functional of the empirical CDF and has a non-normal limiting distribution. While \blb\ is overly conservative and \privBoot\ remains overly conservative, \privSub\ stays close to the non-private subsampling baseline and yields substantially tighter CIs than the competing methods.

These behaviors reflect finite-sample limitations of the two competing resampling approaches. For \blb, bootstrap samples of size $n$ are drawn from subsamples of size $n/s$, so each data point appears approximately $s$ times in a bootstrap replicate. When $n/s$ is small, this can underestimate variability and distort the resulting CI, even though the effect is asymptotically negligible under appropriate choices of $s$. For \privBoot, the main bottleneck is privacy composition: each bootstrap replicate incurs privacy cost, limiting the number of resamples that can be used at a fixed privacy budget. As a result, the finite-sample regimes considered here are far from the asymptotic regimes in which the theoretical guarantees of these methods become accurate.

{\bf Further discussion and empirical evaluations.}
Appendix~\ref{app:detailed_analysis} provides additional empirical evaluations and discussion. Specifically, we include a further comparison between \privSub\ and \privBoot\ under approximate DP (Figure~\ref{fig:ks_gaussian_noise}). We test the robustness of our results by experimenting with the setup of Section~\ref{sec:num_study} for different privacy budgets ($\varepsilon=2,5,8$, Figures~\ref{fig:median_comparison_normal}-\ref{fig:ks_comparison_uniform}), subset size $m$ (Figure~\ref{fig:varying_m}), number of subsets $T$ (Figure~\ref{fig:varying_T}), splits of the privacy budget (Figure~\ref{fig:varying_epsilon_split}), and confidence level (Figures~\ref{fig:varying_significance_width} and \ref{fig:varying_significance_coverage}). 
We conclude with an illustration of the convergence of the empirical CDF of \privSub\ to its limiting distribution, as established in Theorem~\ref{thm:cons_priv_quan_CI} (Figure~\ref{fig:cdf}).

\section{Discussion}

We propose a general subsampling-based framework for constructing differentially private nonparametric confidence intervals. The main advantage of the method is its modularity: under mild conditions, it can be applied to any private estimator as a black box. Its validity relies only on the existence of a limiting distribution and a known convergence rate for the underlying statistic, rather than on normality or smoothness assumptions on the target functional. This makes the procedure applicable in settings where bootstrap-based or normal-approximation methods fail.

Our theoretical results show that, under minimal distributional assumptions, the private subsampling distribution consistently approximates the sampling distribution of the target statistic, leading to asymptotically valid and tight CIs. The empirical results suggest that these theoretical guarantees translate into practical performance: across median estimation, logistic regression, and KS-distance estimation, the method achieves reliable coverage and competitive, often smaller, widths compared with existing general-purpose DP CI methods.

Several directions remain open. As in classical subsampling, the method requires the convergence rate used for rescaling; it could be interesting to develop data-driven methods for estimating the appropriate convergence rate. Similarly, sharper finite-sample guidance for choosing the subset size $m$ and the number of subsets $T$ could further improve practical performance. While the procedure requires repeated evaluation of private estimators on sampled subsets, in practice the dominant constraint is typically privacy composition rather than computational cost. Developing principled methods for balancing these tradeoffs is an interesting direction for future work, especially for complex private estimators and nonstandard statistical targets.

\begin{ack}
We are grateful to Karan Chadha, John Duchi, and Rohith Kuditipudi for helpful conversations regarding their paper \emph{Resampling methods for private statistical inference} \cite{chadha2024resampling}. We also thank Joerg Drechsler, Ira Globus-Harris, Audra McMillan, Jayshree Sarathy, and Adam Smith for discussions about their paper \emph{nonparametric Differentially Private Confidence Intervals for the Median} \cite{drechsler2022nonparametric}. We appreciate that the authors of both papers shared code accompanying their work and generously spent time with us to help us better understand their work and run their code. Any remaining errors are our own.

This work was supported in part by Simons Foundation Collaboration 733792, Israel Science Foundation (ISF) grant 2861/20, a gift from Apple, a grant from the Israeli Council of Higher Education, and ERC grant 101125913. In addition, Shenfeld was supported in part by an Apple Scholars in AIML fellowship. Views and opinions expressed are however those of the author(s) only and do not necessarily reflect those of the European Union or the European Research Council Executive Agency. Neither the European Union nor the granting authority can be held responsible for them.
\end{ack}

\newpage
\bibliographystyle{unsrtnat}
\bibliography{Main_bib}


\clearpage
\appendix

\section{Differential privacy}\label{app:DP}

Let $\datadom$ be an abstract data domain. A dataset of size $n$ is a collection of $n$ individuals' data records: $\dataset=\{\dataset_i\}_{i=1}^n \in \datadom^n$. 
We assume that $n$ is public; that is, we do not protect the size of the dataset.
We call two datasets $\dataset, \dataset' \in \datadom^n$ neighbors, denoted by $\dataset \sim \dataset'$, if they are identical except in one of their records. 

\begin{definition}[Differential privacy]
        Given $\varepsilon \ge 0$, $\delta \in [0,1]$, a data domain $\datadom$ and some domain of responses $\mathcal{R}$, we say that a mechanism ${\cal M}:\datadom^n \rightarrow \mathcal{R}$ satisfies $(\varepsilon, \delta)$-\textit{differential privacy}, denoted by  $(\varepsilon, \delta)$-DP, if 
        $$\prob({\cal M}(\dataset) \in E)\le e^\varepsilon \prob({\cal M}(\dataset') \in E)+\delta $$
        for all $\dataset \sim \dataset' \in \datadom^n$ and all $E \subseteq \mathcal{R}$.
        When $\delta = 0$ we say $\cal M$ is \emph{pure} DP, and if $\delta > 0$ $\cal M$ is said to be \emph{approximate} DP.
\end{definition} 

We detail a few properties of differential privacy that we use throughout the paper. The first is composition.

\begin{lemma}[Basic composition, see Theorem 3.14 in \cite{dwork2014algorithmic}]\label{lem:basic_com}
Let ${\cal M}_1, \ldots, {\cal M}_k$ be randomized algorithms, where each ${\cal M}_i$ satisfies $(\varepsilon, \delta)$-differential privacy. Then the sequence of algorithms $({\cal M}_1, \ldots, {\cal M}_k)$ satisfies $(k\varepsilon, k\delta)$-differential privacy.
\end{lemma}
A more delicate analysis gives a better asymptotic dependence in $\varepsilon$, that is, it composes like the square root of the compositions, but it comes with a non-zero 'failure probability' $
\delta$.
\begin{lemma}[Advanced composition, \cite{dwork2010boosting}]\label{lem:adv_com}
Let ${\cal M}_1, \ldots, {\cal M}_k$ be randomized algorithms, where each ${\cal M}_i$ satisfies $(\varepsilon, \delta)$-differential privacy. Then for any $\delta' > 0$, the sequence of algorithms $({\cal M}_1, \ldots, {\cal M}_k)$ satisfies $(\varepsilon', k\delta + \delta')$-differential privacy, where
$$
\varepsilon' = \varepsilon \left( \sqrt{2k \log(1/\delta')}  + k \left(\frac{e^{\varepsilon} - 1}{e^{\varepsilon} + 1}\right) \right).
$$
\end{lemma}

In this paper, we leverage the privacy amplification given by the subsampling procedure. Luckily, we have very explicit and tight results for sampling without replacement.

\begin{lemma}[Privacy amplification by subsampling, Theorem 9 in {\citet{balle2018privacy}}]\label{lem:amp_by_sub}
Let $\mathcal{M}$ be an $(\varepsilon, \delta)$-differentially private mechanism. Let $\mathcal{M}'$ be the mechanism that, given a dataset of size $n$, selects a subset of $m$ individuals uniformly at random \emph{without replacement}, and applies $\mathcal{M}$ to that subset. Then $\mathcal{M}'$ satisfies $(\varepsilon', \delta')$-differential privacy, where:
$$
\varepsilon' = \log\left(1 + \frac{m}{n} \left(e^\varepsilon - 1\right)\right), \quad \delta' = \frac{m}{n} \cdot \delta.
$$
\end{lemma}

\subsection{Noise addition mechanisms}

One way to achieve DP for algorithms that output numbers (or vectors of numbers) is by noise addition mechanisms. In order to define them, we first define a quantity that is called \textit{Global Sensitivity} see Definition \ref{def:global_sen}, that measures the maximal change an output of a query can change (in some norm), when we change one individual, for any neighboring datasets.

\begin{definition}[Global sensitivity]\label{def:global_sen}
    Given a data domain $\datadom$ and a function $f: \datadom^n\rightarrow \mathbb{R}$, the global sensitivity of $f$ is given by \[\Delta_{f} = \underset{\substack{\dataset,\dataset'\in \datadom^n \\ 
    \dataset \sim \dataset'}}{\max}|f(\dataset)-f(\dataset')|.\]
\end{definition}

The Laplace mechanism (see Definition \ref{Def:Lap_mec}) is one of the classic methods to obtain $(\varepsilon,0)$-DP. Simply put, it adds zero-mean Laplace noise to a statistic, with variance proportional to the global sensitivity. 

\begin{definition}[Laplace mechanism]\label{Def:Lap_mec}
   Consider a data domain $\datadom$ and a function $f: \datadom^n\rightarrow \mathbb{R}$. The Laplace mechanism, denoted by ${\cal M}^{Lap}_f$, simply adds independent Laplace noise to the results of $f$ on a dataset; that is,
    $$ {\cal M}^{Lap}_f(\dataset) = f(\dataset) + Y, \text{ where } Y \sim \operatorname{Lap}(b), \dataset \in \datadom^n,$$
   where $\operatorname{Lap}(b)$ denotes a distribution with probability density function $ p(x) = \frac{1}{2b}\exp\left(-\frac{|x|}{b}\right).$ 
\end{definition}

\begin{lemma}[Theorem 3.6 in \cite{dwork2014algorithmic}]\label{lem:Lap_DP}
    Given some $\varepsilon>0$, the Laplace Mechanism with $b=\Delta_f/\varepsilon$ is $(\varepsilon,0)$-DP, where $\Delta_f$ is the global sensitivity of $f$ (see Definition \ref{def:global_sen}).
\end{lemma} 

The Gaussian Mechanism (Definition \ref{def:Gaus_mec}), similar to the Laplace Mechanism, adds zero-mean Gaussian noise to a statistic, with variance proportional to the global sensitivity.

\begin{definition}[Gaussian mechanism]\label{def:Gaus_mec}
    Consider a data domain $\datadom$ and a function $f: \datadom^n\rightarrow \mathbb{R}$. The Gaussian mechanism, denoted by ${\cal M}^{Gaus}_f$, simply adds independent Gaussian noise to the results of $f$ on a dataset; that is,
    $$ {\cal M}^{Gaus}_f(\dataset) = f(\dataset) + Y, \text{ where } Y \sim \text{N}(0, \sigma^2), \ \dataset \in \datadom^n.$$
\end{definition}

We have the following Theorem given in \cite{dwork2014algorithmic} that establishes privacy guarantees for the Gaussian mechanism

\begin{lemma}[Theorem A.1. in \cite{dwork2014algorithmic}]\label{lem:Gaus_mec_dwork}
   Let $f: \datadom^n \rightarrow \mathbb{R}$ be a function with global sensitivity $\Delta_f$. For any $\varepsilon \in (0,1)$, $\delta \in (0,1)$, and $\dataset \in \datadom$, the Gaussian mechanism ${\cal M}_f^{Gaus}(\dataset)=f(\dataset)+Z$ where $Z\sim \N(0, \sigma^2)$ is $(\varepsilon,\delta)$-DP with 
   $\sigma^2 = \frac{2\ln(1.25/\delta)\Delta_f^2}{\varepsilon^2}.$
\end{lemma} 

\subsection{The exponential and inverse sensitivity mechanisms}

We define the exponential mechanism for a discrete response space, but this definition can easily be generalized for a continuous response space.

\begin{definition}[Exponential mechanism]\label{def:exp_mech}
Fix a privacy parameter $\varepsilon\in\mathbb{R}_+$, and let $\mathcal{R}$ be a finite response space. A mechanism $\mathcal{M}:\datadom^n \to \mathcal{R}$ is a randomized algorithm given by
\begin{equation}
\forall x \in \datadom^n, \forall r \in \mathcal{R}, \quad 
\mathbb{P}[\mathcal{M}(x) = r] = \frac{\exp\!\left(-\tfrac{\varepsilon}{2\Delta}\,\ell(r, x)\right)}{\sum_{r' \in \mathcal{R}} \exp\!\left(-\tfrac{\varepsilon}{2\Delta}\,\ell(r', x)\right)},
\label{eq:expmech}
\end{equation}
where $\Delta$ is the sensitivity of the loss function $\ell:\mathcal{R}\times\datadom^n\to\mathbb{R}$ given by
\begin{equation}
\Delta = \sup_{x,x' \in \datadom^n : d(x,x') \leq 1} \; \max_{r \in \mathcal{R}} 
\big|\ell(r,x) - \ell(r,x')\big|,
\label{eq:sensitivity}
\end{equation}
where the supremum is taken over all datasets $x$ and $x'$ that differ on the data of a single individual (which we denote by $d(x,x') \leq 1$).

\end{definition}

\begin{lemma}[Privacy of the exponential mechanism, see Theorem 6 in \cite{mcsherry2007mechanism}]\label{lem:exp_mech_privacy}

Let $\mathcal{M}$ be as in Definition~\ref{def:exp_mech} with loss function $\ell:\mathcal{R}\times\datadom^n\to\mathbb{R}$ and sensitivity $\Delta$ defined in~\eqref{eq:sensitivity}. Then $\mathcal{M}$ satisfies $(\varepsilon,0)$-differential privacy. 
\end{lemma}

This mechanism was implicitly considered by \citet{mcsherry2007mechanism} and formally introduced by \citet{asi2020instance}.

\begin{definition}[Inverse sensitivity for median]\label{def:inv_sen}
Fix a privacy parameter $\varepsilon\in\mathbb{R}_+$ and let $f:\datadom^n\!\to\!\mathcal{R}$ and let $d(\cdot,\cdot)$ be the dataset metric
underlying the adjacency $\dataset\sim \dataset'$ (e.g., Hamming). For $\dataset\in\datadom^n$ and
$r\in\mathcal{R}$, define the \emph{inverse sensitivity}
\[
\operatorname{len}_f(\dataset;r)\;\triangleq\;
\inf\{\, d(\dataset,\dataset') : \dataset'\in\datadom^n,\ f(\dataset')=r \,\},
\]

The \emph{inverse sensitivity mechanism} draws $R\in\operatorname{im}(f)$ with
\[
\prob{\big(R=r\mid \dataset\big)}\;\propto\;
\exp\!\left(-\tfrac{\varepsilon}{2}\,\operatorname{len}_f(\dataset;r)\right),
\]

In the case of median estimation, given a dataset 
\(x \in \mathbb{R}^n\), our goal is to compute \(\operatorname{Median}(x)\). 
For simplicity, we assume \(x_i \in [0,R]\) for some \(R > 0\). 
The theory and derivations remain unchanged if the data are unbounded, 
in which case we redefine
\[
f(x) = \min\{R, \max\{-R, \operatorname{Median}(x)\}\}.
\]

To implement mechanism~(M.2), we compute \(\operatorname{len}_f\). 
Let \(m = \operatorname{Median}(x)\). Then, for \(t \in [0,R]\),
\[
\operatorname{len}_f(x; t) 
= \bigl|\{\, x_i : x_i \in (t,m] \,\cup\, [m,t) \,\}\bigr|.
\]

\end{definition}

\begin{lemma}[Privacy of the inverse sensitivity mechanism, see Lemma 3.1 in \citet{asi2020instance}]
    Let $\operatorname{len}_f(\dataset;r)$ be as in Definition~\ref{def:inv_sen}.
The mechanism that outputs $Y\in\operatorname{im}(f)$ with
\[
\prob\big(R=r\mid \dataset\big)\;\propto\;\exp\!\left(-\tfrac{\varepsilon}{2}\,\operatorname{len}_f(\dataset;r)\right)
\] is $(\varepsilon,0)$-DP by Lemma~\ref{lem:exp_mech_privacy}.\end{lemma}

\clearpage
\section{Complementary claims and missing proofs}\label{app:miss_proof}
\subsection{Consistency of private mechanisms}
The requirements of $\tau_n$ consistency require that, with respect to the underlying distribution, the difference between the private and non-private estimators, rescaled by $\tau_n$, goes to zero in probability. In this section, we prove something stronger since the bounds on the privacy-preserving mechanisms hold uniformly over datasets; the same inequality holds for any underlying distribution. 

\begin{claim}[Noise addition vanishes at rate $\tau_n$]\label{clm:lap_vanish}
Define ${\cal M}(\boldsymbol{\omega})=\theta(\dataset)+N_n$, where $N_n$ is independent of the data. If for any $\eta>0$,
$\prob(\tau_n \cdot |N_n| > \eta) \rightarrow 0$,
then $\tau_n({\cal M}(\dataset)-\theta(\dataset))\pto 0$.
\end{claim}

\begin{proof}
For any fixed $\eta>0$,
$$
\prob\big(\tau_n \cdot |{\cal M}(\dataset)-\theta(\dataset)| >\eta\big) = \prob\big(|\tau_n \cdot N_n|>\eta\big) \pto 0,
$$
so ${\cal M}$ is $\tau_n$-consistent. 
\end{proof}

The condition of Claim~\ref{clm:lap_vanish} holds for the standard DP noise additions:

\textit{Laplace.} If $N_n\sim\mathrm{Lap}(0,b_n)$, then for any $\eta>0$,
$$
\prob\big(|\tau_n N_n|>\eta\big)=\exp\!\left(-\frac{\eta}{\tau_n b_n}\right),
$$
so the claim applies whenever $\tau_n \cdot b_n\to 0$.

\textit{Gaussian.} If $N_n\sim N(0,\sigma_n^2)$, then for any $\eta>0$,
$$
\prob\big(|\tau_n N_n|>\eta\big)\le 2\exp\!\left(-\frac{\eta^2}{2\tau_n^2\sigma_n^2}\right),
$$
hence the claim applies whenever $\tau_n \cdot \sigma_n\to 0$.

For a concrete example, consider mean queries, where the data is bounded, w.l.o.g in [0,1]. The mean is normally distributed at rate $\tau_n=\sqrt{n}$ from the central limit theorem. For the Laplace noise addition, if we set $b_n = \frac{1}{n\varepsilon}$, the mechanism is $\varepsilon$-DP, and as long as $\sqrt{n} \cdot \varepsilon \rightarrow \infty$, it is $\tau_n$-consistent. A similar example can be derived for Gaussian noise addition.

\begin{remark}[Subsampling and composition]
For a subsample $I\subset [n]$, $|I|=m$, consistency requires $\varepsilon\sqrt{m}\to\infty$, 
where $\varepsilon$ accounts for composition (over $T$ releases) and amplification by subsampling. For example, settings with 
$m\sqrt T=O(n)$ can keep $\varepsilon=O(1)$ and hence preserve consistency.
\end{remark}

Before we state the following claim, we introduce a mathematical notation. Let $X_n$ be a random variable and $a_n>0$. We write $X_n=O_p(a_n)$ if for every $\eta>0$ there exist $\xi<\infty$ and $N$ such that
$$
\prob\big(|X_n|>\xi a_n\big)<\eta \quad \text{for all } n>N.
$$

\begin{claim}[$\tau_n$-consistency of the inverse sensitivity mechanism for median estimation]\label{clm:expmech_median_fixed}
    Under the conditions of Proposition 5.1 in \cite{asi2020instance}, the inverse sensitivity mechanism for median is $\tau_n$ consistent as long as $\frac{\log(n)}{\sqrt{n}\varepsilon} \to 0$.
\end{claim}

Their conditions, generally speaking, require continuous density around the true median of the underlying distribution. The proof follows immediately below the proposition, where they show that as long as $\frac{\log(n)}{\sqrt{n}\varepsilon} \pto 0$, under a certain choice of the parameters, the mechanism satisfy $\tau_n(\theta(\dataset)-{\cal M}(\dataset)) \to 0$, that is, it is $\tau_n$ consistent.

\subsection{Accuracy of CIs}
\begin{lemma}[Valid quantile-based CI]\label{lem:val_CI_emp_dist}
Let $\widehat{\theta}_n=\theta(\dataset)$ with $\dataset\sim P^{(n)}$, where $\tau_n(\widehat{\theta}_n-\theta^*)\dto U$ for a CDF $U$. Let $\hat{V}_n$ be a \emph{random CDF}, a CDF-valued statistic: for each $x\in\R$, $\hat{V}_n(x)\in[0,1]$ is a random variable (randomness from the data, and any additional randomized procedure such as subsampling indices or a privacy mechanism). All probabilities below are unconditional over these sources.

Assume that for every continuity point $x$ of $U$, $ \hat{V}_n(x)\ \pto\ U(x)$. Fix $\alpha\in(0,1)$ and define the empirical quantile $\hat q_{n}(\alpha)\coloneqq \inf\{x:\hat{V}_n(x) \geq \alpha\}$.
If $U(\cdot)$ is continuous at its $\alpha/2$ and $1-\alpha/2$ quantiles, then the interval
$$
\Big[\widehat{\theta}_n-\tau_n^{-1}\hat q_{n}(1-\alpha/2),\ \widehat{\theta}_n-\tau_n^{-1}\hat q_{n}(\alpha/2)\Big]
$$
is asymptotically valid and tight for the quantity $\theta^*$.
\end{lemma}

\begin{proof}
    This is a direct result of the method of proof of Theorem 2.2.1 in \cite{politis1999subsampling}, where they prove a one-sided CI.    
\end{proof}

\subsection{Proof of Theorem \ref{thm:cons_priv_quan_CI}}
\begin{theorem}[Theorem \ref{thm:cons_priv_quan_CI}, restated]
The results of Theorem \ref{thm:non_priv_ci_consis} hold when replacing $\widehat{U}_{n,m}(x)$ by $\widetilde U^{T}_{n,m}(x)$ as long as the perturbation mechanism is $\tau_n$-consistent (Definition \ref{def:tau_n_cons}), and $T \rightarrow \infty$.
\end{theorem}

\begin{proof}
We define a family of jointly distributed random variables, $\widehat{\theta}_n$, $\widetilde{\theta}_n$ $\widehat{\theta}_{n,i}$ $\widetilde{\theta}_{n,i}, \ i=1,\ldots,T$ by the following procedure:
We first sample a dataset $\dataset \sim P^{(n)}$. Then, given a subset size $m \in [n-1]$, we sample a list $\boldsymbol{I}=(I_1,\ldots,I_T)$ of subsets of indices, $I_i\subset [n]$, $|I_i|=m$ uniformly with replacement (over the sets). Lastly, we compute
$$ \widehat{\theta}_n \coloneqq \theta(\dataset), 
\ \ \widehat{\theta}_{n,i}\coloneqq \theta(\dataset(I_i)), \ \ \widetilde{\theta}_n \coloneqq {\cal M}(\dataset), \ \ \widetilde{\theta}_{n,i}\coloneqq {\cal M}(\dataset(I_i)), \ \ i=1,\ldots,T,$$
where the dataset is the same for all random variables, but the perturbation is independent of the data and of the different random variables.

Note that we define $\widehat{U}^T_{n,m}$ and $\widetilde{U}^T_{n,m}$ by these random variables, simply taking
$$ \widehat{U}^T_{n,m}(x)  = \frac{1}{T} \sum_{i=1}^T \mathbbm{1} \{\tau_m \cdot (\widehat{\theta}_{n,i}-\widehat{\theta}_{n}) \leq x \}, \quad
\widetilde{U}^T_{n,m}(x)  = \frac{1}{T} \sum_{i=1}^T \mathbbm{1} \{ \tau_m \cdot (\widetilde{\theta}_{n,i}-\widetilde{\theta}_{n}) \leq x \}$$

Fix some $r\in \R$, we can write:
\begin{multline}\label{eq:deco_priv_est}
    \widetilde U^{T}_{n,m}(x_0)-U(x_0) =  
    \underbrace{\Big( \widetilde U^{T}_{n,m}(x_0)-\widehat U^{T}_{n,m}(x_0+2r) \Big)}_{\text{(a)}}
    + \underbrace{\Big(\widehat U^{T}_{n,m}(x_0+2r)- U(x_0+2r)\Big)}_{\text{(b)}}\\
    + \underbrace{\Big(U(x_0+2r) - U(x_0)\Big)}_{\text{(c)}}
\end{multline}

Since the number of discontinuity points of a CDF is countable, given any continuity point $x_{0}$ of $U$ there exists two sequences indexed by $k\in \N$,  $r_{k}^{+}\downarrow 0$ and $r_{k}^{-}\downarrow 0$, where $U(\cdot)$ is continuous at $x_0+2 r_{k}^{+}$, $x_0-2 r_{k}^{-}$ for all $k$. We will use these sequences to prove all three RHS term converge to $0$ ($(a)$ and $(b)$ in probability, and $(c)$ deterministically), and from Slutsky's theorem, conclude that $\widetilde U^{T}_{n,m}(x_0)-U(x_0) \pto 0$. 

The fact $(c) \pto 0$ is a direct result of the fact $x_0$ is a continuity point, so we focus on the other two terms. In the following sections, inequalities between random variables below should be interpreted pointwise: they hold for every realization except on events of probability zero.

\paragraph{Proof that $\boldsymbol{(a) \pto 0}$}

For the term, $\Big( \widetilde U^{T}_{n,m}(x_0)-\widehat U^{T}_{n,m}(x_0+2r) \Big)$, we define the following random variables:

$$
\widehat\Delta_{n,i}\coloneqq \widehat{\theta}_{n,i}-\widetilde{\theta}_{n,i},
\qquad
\widehat\Delta_n \coloneqq \widehat{\theta}_n-\widetilde{\theta}_n.
$$

Note that, for each $i$, we can write the following
$$ \tau_m \cdot (\widetilde\theta_{n,i}-\widetilde\theta_n) = \tau_m \cdot (\widehat\theta_{n,i}-\widehat\theta_n) - \tau_m \cdot \widehat\Delta_{n,i} + \tau_m \cdot \widehat\Delta_n
$$
since they are all defined over the same draw $\dataset  \sim P^{(n)}$, and $I_i$.

Next we notice that for any $r \ge 0$
\[
    \mathbbm{1}\{\tau_m \cdot (\widetilde\theta_{n,i}-\widetilde\theta_n)\le x\}\ - \mathbbm{1}\{\tau_m \cdot (\widehat\theta_{n,i}-\widehat\theta_n) \le x+2r\} \le \mathbbm{1}\{\tau_m \cdot |\widehat\Delta_{n,i}|\ge r\} + \mathbbm{1}\{\tau_m \cdot |\widehat\Delta_n|\ge r\},
\]
so denoting $\hat \phi_{n, m, T}(r) \coloneqq \frac{1}{T}\sum_{i=1}^T \mathbbm{1}\{\tau_m \cdot |\widehat\Delta_{n,i}|\ge r\}$, $\hat \psi_{n, m, T}(r) \coloneqq \mathbbm{1}\{\tau_m \cdot |\widehat\Delta_n|\ge r\}$ we have

\begin{equation}\begin{split}\label{eq:bound_priv_by_non+res}
     \widetilde{U}^T_{n,m}(x) - \widehat{U}^T_{n,m}(x+2r)  &= \frac{1}{T} \sum_{i=1}^T \mathbbm{1} \{ \tau_m \cdot (\widetilde{\theta}_{n,i}-\widetilde{\theta}_{n}) \leq x\} - \frac{1}{T} \sum_{i=1}^T \mathbbm{1} \{ \tau_m \cdot (\widehat{\theta}_{n,i}-\widehat{\theta}_{n}) \leq x+2r \} \\
     & \le \frac{1}{T} \sum_{i=1}^T \left(\mathbbm{1}\{\tau_m \cdot |\widehat\Delta_{n,i}|\ge r\} + \mathbbm{1}\{\tau_m \cdot |\widehat\Delta_n|\ge r\} \right)  \\
     &\leq\hat{\phi}_{n, m, T}(r) + \hat{\psi}_{n, m, T}(r).
 \end{split}\end{equation}
Intuitively, $\hat{\phi}_{n, m, T}$ measures the fraction of subsamples where the DP perturbation exceeds the tolerance $r$, in other words, it counts the fraction of ``bad events'', and $\hat{\psi}_{n, m, T}$ indicates the same for the full sample estimate.

An analogous bound with $x-2r$ yields
\begin{equation}\label{eq:bound_priv_by_non+res_side}
     \widetilde{U}^T_{n,m}(x) - \widehat{U}^T_{n,m}(x-2r) 
     \geq -(\hat{\phi}_{n, m, T}(r) + \hat{\psi}_{n, m, T}(r))
\end{equation}

Note that since index sequences are sampled iid, we have that 
$$\E[\hat{\phi}_{n, m, T}(r)]=\prob(\tau_m \cdot |\widehat\Delta_{n,1}|\ge r) = \prob(\tau_m \cdot |\widehat\Delta_{m}|\ge r)$$
and 
$$\E[\hat{\psi}_{n, m, T}(r)] = \prob(\tau_m \cdot |\widehat\Delta_n|\ge r) \le \prob(\tau_n \cdot |\widehat \Delta_n|\ge r)$$
where expectation is taken over all sources of randomness. 

Since $\tau_m \cdot \widehat\Delta_m \pto 0$ and $\tau_n \cdot \widehat\Delta_n \pto 0$, for each $k$ there exists some $N_{1}(k)$ such that for any $n \ge N_{1}(k)$ we have (note that $m$ is a function of $n$),
\[
    \max\{\prob(\tau_m \cdot |\widehat\Delta_{m}|\ge r_{k}^{+}), \prob(\tau_n \cdot |\widehat \Delta_n|\ge r_{k}^{+}), \prob(\tau_m \cdot |\widehat\Delta_{m}|\ge r_{k}^{-}), \prob(\tau_n \cdot |\widehat \Delta_n|\ge r_{k}^{-})\} \le 1/k,
\]
so
\[
    \E[\hat{\phi}_{n, m, T}(r_{k}^{+}) + \hat{\psi}_{n, m, T}(r_{k}^{+})] \to 0 ~~~~\text{and}~~~~\E[\hat{\phi}_{n, m, T}(r_{k}^{-}) + \hat{\psi}_{n, m, T}(r_{k}^{-})] \to 0.
\]

Using Markov's inequality and the fact that $\hat{\phi}_{n, m, T}, \hat{\psi}_{n, m, T}$ are bounded non negative random variables, this implies
\[
    \hat{\phi}_{n, m, T}(r_{k}^{+}) + \hat{\psi}_{n, m, T}(r_{k}^{+}) \pto 0 ~~~~\text{and}~~~~  \hat{\phi}_{n, m, T}(r_{k}^{-}) + \hat{\psi}_{n, m, T}(r_{k}^{-}) \pto 0,
\]
which completes the proof that $(a) \pto 0$.

\paragraph{Proof that $\boldsymbol{(b)\pto 0}$}
Since $x_{0} + 2 r_{k}^{+}$ and $x_{0} - 2 r_{k}^{-}$ are all continuity points, by Theorem ~\ref{thm:non_priv_ci_consis}, under standard subsampling setting, for each $k$ there exists some $N_{2}(k)$ such that for any $n \ge N_{2}(k)$ we have (note that $m, T$ are a function of $n$),
$$
\prob\!\left(\big|\widehat U^{T}_{n,m}(x_0+2r_{k}^{+})-U(x_0+2r_{k}^{+})\big|>\tfrac{1}{k}\right)\ \le\ \tfrac{1}{k},
$$
and
$$ \prob\!\left(\big|\widehat U^{T}_{n,m}(x_0-2r_{k}^{-})-U(x_0-2r_{k}^{-})\big|>\tfrac{1}{k}\right)\ \le\ \tfrac{1}{k},$$
which implies
$$
\widehat U^{T}_{n,m}(x_0+2r_{k}^{+})-U(x_0+2r_{k}^{+})\ \pto\ 0 ~~~~\text{and}~~~~  \widehat U^{T}_{n,m}(x_0-2r_{k}^{-})-U(x_0-2r_{k}^{-})\ \pto\ 0 .
$$

\paragraph{Uniform convergence in probability}
The proof of the second part follows the same path, and from the continuity assumption can simply set $r_{n} = 1/n$. 

From \eqref{eq:deco_priv_est}, using the triangle inequality and taking supremum over $x \in \R$,
\begin{multline}\label{eq:sup_bound}
\sup_{x}\big|\widetilde U^{T}_{n,m}(x)-U(x)\big|
\ \le\
\underbrace{\sup_{x\in\R}\big| \widetilde U^{T}_{n,m}(x)-\widehat U^{T}_{n,m}(x+2r) \big|}_{\text{(a)}}
\ +\
\underbrace{\sup_{x\in\R}\big|\widehat U^{T}_{n,m}(x+2r)- U(x+2r)\big|}_{\text{(b)}}
\ \\ +\
\underbrace{\sup_{x\in\R}\,|U(x+2r)-U(x)|}_{\text{(c)}}.
\end{multline}

The proof that $(a) \pto 0$ directly applies since it did not depend on $x_{0}$. We note that for any $r$ we have
$$ \sup_{x\in\R}\big|\widehat U^{T}_{n,m}(x+2r)- U(x+2r)\big| = \sup_{x\in\R}\big|\widehat U^{T}_{n,m}(x)- U(x)\big|,$$
so Theorem \ref{thm:non_priv_ci_consis} directly implies that under the standard subsampling setting with $T\to\infty$ we have $(b) \pto 0$.

Finally $(c) \pto 0$ from the definition of uniform continuity which completes the proof. 

\end{proof}

\newpage
\section{More related work}

There is a rich literature on DP estimation, giving both asymptotic and finite-sample accuracy guarantees for various statistics such as \emph{mean} and \emph{moments} \citep{smith2008efficient, bun2019average}, \emph{quantiles} \citep{nissim2007smooth,gillenwater2021differentially,kaplan2022differentially,durfee2023unbounded, asi2020instance}, \emph{covariance estimation} and \emph{PCA} \citep{hardt2014noisy,dong2022differentially}, broad classes of \emph{ratios} and other \emph{M-estimators} \citep{lei2011differentially, smith2011privacy, shoham2025differentially}, and \emph{linear and logistic regression} \citep{chaudhuri2011differentially,zhang2012functional, wang2018revisiting}. 

However, constructing DP CIs for statistics is a more complex task.
In parametric settings where CIs can be derived from distribution parameters, this can be done using a private estimation of these parameters~\citep{du2020differentially, karwa2018finite, shoham2022asking}. In the nonparametric setting, however, where minimal assumptions (if any) are made about the distribution (e.g., bounded range or moments), such approaches are not applicable. If the distribution of the statistic approaches a \emph{known} parametric limiting distribution (e.g., a normal distribution), CIs can be constructed by estimating its parameters, which can also be done privately \citep{wang2018statistical}. In the absence of a convenient limiting distribution, CIs can sometimes be constructed for specific quantities when their CI can be expressed as another parameter that can be empirically estimated. For example, \citet{drechsler2022nonparametric} provides nonparametric DP CIs for the median by directly estimating other quantiles.
 
The more general resampling methods, such as bootstrapping, imply each element might participate in the resampled dataset more than once, increasing the query's sensitivity to changing a single element, adversely affecting the privacy-accuracy tradeoff. 
\citet{brawner2018bootstrap} tackle this by noting that the maximal number of times each element can be sampled is very low with high probability, so capping the maximal number of appearances gives statistical guarantees that are nearly identical to classical bootstrapping. 

\citet{wang2025differentially} further refine the privacy analysis by characterizing the privacy loss of a single bootstrap sample and its composition using the $f$-DP framework (which translates to approximate-DP). However, their validity guarantees are restricted to the Gaussian mechanism and discrete underlying distributions (though these limitations might be an artifact of the analysis), and in practice, a numerical privacy profile can be provided only for Gaussian DP mechanisms. While the relative degradation in CI width caused by bootstrap privatization decreases as the number of bootstrap samples$\rightarrow \infty$, we show empirically that constant factors play a crucial role.
That said, if the distribution of the noise is independent of the data, which is the case with the noise addition mechanism, one can deconvolute the privacy perturbation noise from the resulting CDF, deriving a valid CI that quantifies the uncertainty without accounting for the extra perturbation. This is an empirical method (no formal proof of validity is provided), but empirically, they show it is very efficient.  

In recent years, a technique known as Bag-of-Little-Bootstrap (\textsc{BLB}), proposed by \citet{kleiner2014scalable}, has emerged as a valuable tool for non-private bootstrapping of large databases. It relies on splitting the data into (disjoint) subsets. For each subset, one bootstraps many samples of the original size and computes a CDF/standard deviation. Finally, one aggregates to a final estimation. BLB subset estimates can also be aggregated \emph{privately} using the subsample-and-aggregate technique \citep{nissim2007smooth}. 

The first paper using this technique was by \citep{covington2025unbiased}, where they used the CoinPress mechanism to aggregate the results, and proved the CI are valid asymptotically. That said, the coin press mechanism as an aggregator performs poorly for small sample sizes, making their method empirically limited for moderate sample sizes. Elegantly, the subsample-and-aggregate technique does not generally require a private estimator of the target quantity, but only of the aggregation, but yet, their analysis is aimed solely for a normal limiting distribution.

The work of \citep{chadha2024resampling} aimed to derive rates for the BLB framework. They worked under much stronger assumptions (Edgeworth expansion), and they also required the existence of a privacy-preserving mechanism to release the quantity of interest. Similarly, they also assume a normal limiting distribution. They proposed two aggregation methods - one, for a quantile-based CI (to which we compared in this paper), some variation of the AboveThreshold mechanism, over a decreasing in size intervals that are considered as CIs, or the inverse sensitivity mechanism over a vector of SD estimates.

Above all, the BLB technique uses sample splitting, thus, each split must be sufficiently representative of the underlying distribution, which is hard to achieve with small datasets when the underlying distribution is more challenging.

\newpage
\section{Detailed analysis}\label{app:detailed_analysis}

In this Appendix, we present additional analyses and provide further details of the experiments described in Section~\ref{sec:num_study}. Supplementary figures, extended discussions, and numerical results are included to give a more complete account of the study.

\paragraph{Code and reproducibility.}
Our implementations of the alternative methods are based on the original works: the exponential mechanism (\expMech) from \href{https://github.com/anonymous-conf-medians/dp-medians}{\texttt{Ira Globus-Harris's} Git repository}, \blb\ from  \href{https://github.com/comptastics/dp_inference}{\texttt{knchadha} Git repository}, and \privBoot\ from \href{https://github.com/Zhanyu-Wang/Differentially_Private_Bootstrap}{\texttt{Zhanyu-Wang} Git repository}. We also intend to make our implementation public.

\subsection{Full description of all methods}

we described at length our method throughout the paper. We now describe the other algorithms we compare to, and the choices of hyperparameters.

\subsubsection*{Algorithm \expMech}

The algorithm is described in \cite{drechsler2022nonparametric}. The key idea is to use the statistical property of order statistics in combination with the exponential mechanism. At the core of the approach lies the observation that the median splits the distribution into two equal halves: the probability of any data point being greater than or less than the median is exactly $0.5$. This property allows the construction of nonparametric confidence intervals by looking at the position (rank) of observations in the sorted data, rather than making assumptions about the underlying distribution. To privatize this process, one can use the exponential mechanism, where the ``utility'' corresponds to how well a chosen rank represents the desired quantile.\\
Drechsler and coauthors use a widened version of the exponential mechanism, which balances the tradeoff between errors in the rank domain and errors in the value domain. This is controlled by a parameter called \emph{granularity}, which determines how ranks are translated into intervals. Smaller values emphasize rank accuracy, while larger values favor stability in the value domain. Since overly large granularity may dominate the resulting interval width, in our implementation, we scale it with the sample size, setting it to $0.1/\sqrt{n}$, so that the intervals continue to shrink as $n$ grows.

\subsubsection*{Algorithm \blb}

The algorithm is described in \cite{chadha2024resampling}. The algorithm, similar to \privSub, first estimates the quantity of interest on the full sample in a differentially private way. For example, for median estimation, it uses the inverse sensitivity mechanism. The privacy budget is divided equally between estimating the full sample estimator and the aggregation step. We note that \citet{chadha2024resampling} also consider another algorithm that estimates the standard deviation and construct a normal-approximation CI, which we did not include, because we are interested in the most general setting. 

The dataset is then partitioned into $s$ subsets. For each set, bootstrapping is used to generate $T(n)=\max\{\min\{n^{1.5}/\log n, 1000\}, 100\}$ samples of size $n$, and generate a non-private estimate for each bootstrap sample. Each estimate is centered by the full-sample estimator, and rescaled by $\sqrt{n}$ (they assume normality, so this is the convergence rate). This results in $s$ vectors of $T$ non-private centered and scaled estimates. 
For the aggregation step, a variation of the aboveThreshold is used (See Algorithm 1 in \cite{chadha2024resampling}). A sequence of intervals decreasing at rate $\sqrt{n}$ centered at zero is fixed in advance.  For each vector of non-private estimates, the number of values contained within each interval is recorded. Consequently, for each interval, a list of coverages across the splits is obtained. Using their variant of the \texttt{AboveThreshold} algorithm, the procedure stops once the median coverage falls below the confidence level, and the previous interval is returned. Intuitively, this yields a private and consistent estimate of the estimator's variability; confidence intervals then follow by rescaling and centering at the full-sample estimator. Privacy is guaranteed by composing the budget allocated to the full-sample estimator with that of the aggregation step.

\subsubsection*{Algorithm \privBoot}

The algorithm is described in \cite{wang2025differentially}. It closely mirrors the non-private bootstrap, except that the statistic is replaced by its private counterpart. Confidence intervals are then obtained by post-processing the resulting private empirical CDF, using the quantiles corresponding to the desired confidence level. Unlike \privSub\ and \blb, this approach does not require splitting the privacy budget and estimating the statistic on the full dataset; instead, it relies on the average of the private bootstrap estimates. For a fair comparison, we use the same number of bootstrap samples~$B$ as in our \privSub\ experiments.

\subsubsection*{Non-private bootstrap}

There are many variations of the bootstrap, such as the smoothed bootstrap. In this work, the following version is employed: given a sample of size $n$, $n$ observations are drawn with replacement from the original dataset, and this procedure is repeated $T(n) = \max\{\min\{5 \sqrt{n}, 500\}, 200\}$ times. For each bootstrap sample, the center is estimated, and the statistic of interest is computed. A confidence interval is then obtained from the empirical distribution of these estimates.

\subsection{Full details of the different experiments}

\paragraph{Median estimation.}
For each distribution, sample size $n$, and algorithm, we construct $1000$ nominal $90\%$ CIs, each using an independent dataset of size $n$. We report the average CI width and the empirical coverage, defined as the fraction of intervals containing the true distributional median. We set $T=60$ for both \privSub\ and \privBoot, and use $m=n^{2/3}$ for \privSub. For the median-specific exponential-mechanism baseline \expMech, we use grid granularity $1/(10\sqrt{n})$. For \blb, we set the number of splits to $s=\lceil 10\log(n)/\varepsilon\rceil$, following \cite{chadha2024resampling}. Other method-specific hyperparameters are chosen according to the implementations and recommendations in the corresponding papers.

All median experiments use total privacy budget $\varepsilon=5$ under pure DP, except for \privBoot, which is analyzed under approximate DP and is run with $(\varepsilon=5,\delta=10^{-6})$. For \privSub\ and \blb, when the algorithm requires splitting the privacy budget across multiple private subroutines, we split the budget equally. For the private median estimator used inside the resampling procedures, we use the inverse sensitivity mechanism, which is pure $\varepsilon$-DP. Since the privacy accounting for \privBoot\ is formulated in terms of $\mu$-GDP, we convert pure DP guarantees to GDP using the fact that every $\varepsilon$-DP mechanism is also $\mu$-GDP with
$
\mu=-2\Phi^{-1}\left(1/(1+e^\varepsilon)\right)
$
(Theorem~5.1 in \citealp{liu2022identification}). We note that the performance of \privBoot\ in this setting may be improved by a sharper GDP analysis of the inverse sensitivity mechanism, or by using an alternative DP median estimator whose privacy guarantees are naturally expressed in the GDP framework.

\paragraph{Logistic regression output perturbation.}
We also consider CI construction for the slope parameter in a simple logistic regression model, following the setup of \citet{wang2025differentially}. For each observation, we sample $z_i \sim \operatorname{Unif}[0,1]$ independently and set $x_i=(1,z_i)^\top$, corresponding to an intercept and one covariate. Given a fixed parameter $\theta=(\theta_1,\theta_2)$, we generate binary responses $y_i \in \{-1,1\}$ according to
$
\prob(Y_i=y_i\mid X_i=x_i)
=
\frac{1}{1+\exp(-y_i(\theta)^\top x_i)}.
$
The target is the slope coordinate of the population minimizer of the regularized logistic risk. The parameter $\theta=(0,0.8)$ is used only to generate the data.

For each dataset, the base private estimator is obtained using the output perturbation mechanism for regularized logistic regression from \citet[Algorithm~5]{wang2019differentially}. Specifically, we first compute the regularized ERM
$
\widehat{\theta}_c
=
\arg\min_{\theta\in\mathbb{R}^2}
\left\{
\frac{1}{n}\sum_{i=1}^n -\log \Pr(y_i\mid x_i)
+
c\|\theta\|_2^2
\right\},
$
and then release
$
\widetilde{\theta}
=
\widehat{\theta}_c+\xi$,
$\xi\sim\mathcal{N}(0,\sigma^2 I_2)$.
The noise level $\sigma^2$ is calibrated using the sensitivity bound for the regularized ERM from \citet{wang2019differentially}; since $z_i\in[0,1]$, we have $\|x_i\|_2\leq \sqrt{2}$, yielding sensitivity of order $1/(nc)$. In this experiment, all private methods are run with total privacy budget $(\varepsilon=5,\delta=10^{-6})$. As in the other experiments, we construct $1000$ nominal $90\%$ CIs for each sample size and method, and report empirical coverage and average width. We use $T=60$ for \privSub\ and \privBoot, $m=n^{2/3}$ for \privSub, and the remaining method-specific hyperparameters are chosen according to the corresponding original papers. For the logistic regression, we fix parameters $\theta=(0,0.8)$, and the regularization $c=0.1$.

\paragraph{Kolmogorov--Smirnov statistic.}
Finally, we evaluate CI construction for the Kolmogorov--Smirnov statistic
$
D_n=\sup_x |F_n(x)-F(x)|,
$
where $F_n$ is the empirical CDF and $F$ is the target CDF. In the experiment shown in Figure~\ref{fig:main_fig}, we consider the uniformity-testing setting: the data are sampled as $\omega_i\sim \operatorname{Unif}[0,1]$, and the target distribution is $F(x)=x$ on $[0,1]$. Under this null model, $\sqrt{n}D_n$ converges to the Kolmogorov distribution. This example is useful because the KS statistic has a non-normal limiting distribution (Kolmogorov distribution). In this experiment, we take subsampling as the non-private baseline instead of the bootstrap, because the non-private bootstrap yields undercoverage for most of the sample sizes considered.

The KS statistic has global sensitivity $1/n$ \citep{awan2025differentially}, and can therefore be privatized directly by standard noise-addition mechanisms. For \privSub, we apply the Laplace mechanism to the KS statistic computed on each subsample, with sensitivity $1/m$. For \privBoot, we use the Gaussian mechanism, following the privacy accounting of \citet{wang2025differentially}. As in the median experiment, all pure-DP methods are run with total privacy budget $\varepsilon=5$, while \privBoot\ is run with $(\varepsilon=5,\delta=10^{-6})$. We again set $T=60$ for \privSub\ and \privBoot, use $m=n^{2/3}$ for \privSub, and construct $1000$ nominal $90\%$ CIs for each sample size and method. We report empirical coverage and average CI width.

\subsection{A discussion about the limitations and advantages of different methods}

Selecting a private inference procedure is delicate because performance depends jointly on many factors, such as the regularity of the target functional and data distribution, the availability of accurate privacy mechanisms, how privacy composition interacts with subsampling or resampling, and more. We outline guiding questions and some discussion.

First, we need to ask: Is there a privacy mechanism that accurately privatizes the statistic of interest on the full sample?
If the answer is no, i.e., there is no mechanism that satisfies $\tau_n$-consistency (e.g., ratios and other sensitive statistics), then BLB should be considered for its use of the subsample-and-aggregate technique. The algorithm \blb\ also assumes a variation of $\tau_n$-consistency, but the BLB framework can be modified to be more general.

Is the underlying distribution ``well-behaved'' at the scale of each split?
The accuracy of BLB relies on each split statistically representing the population so that the split-level statistic has (approximately) the same law as the full-sample statistic. This tends to hold for unimodal, light-tailed, continuous distributions, but can fail with heavy tails, mixture structure, or discontinuities, as shown in Figure \ref{fig:main_fig}. Some splits may land in different mixture components or be dominated by tail observations, biasing the aggregate. While the choice $s=O\left( \log(n)/\varepsilon\right)$ in \blb\ implies that these effects are asymptotically negligible, their effect is significant at a reasonable sample size.
Private subsampling and bootstrapping CDF tolerates more heterogeneity. In the subsampling framework, it only requires the classical conditions ($m\to\infty$, $m/n\to 0$) and a weak limit for the root; individual splits need not be ``representative'' in the BLB sense as long as the \emph{empirical} subsampling CDF converges.

Must the method ``know'' the convergence rate $\tau_n$?
Our approach rescales quantiles from the $m$-scale to the $n$-scale, which \emph{uses} the ratio $\tau_m/\tau_n$ (e.g., $\sqrt{m/n}$). BLB and Bootstrap methods do not assume rate knowledge; they use resampling to size $n$, but at a cost: in the case of BLB the statistical validity depends on the bootstrap approximation at the privatized-split level and on the aggregation rule, and for private bootstrap it comes at a cost in constants. 

Does one need \emph{the whole distribution} (not just a single CI)?
\privSub\ and \privBoot\ return an estimate of the \emph{entire} CDF of the estimator (at the $n$-scale). This enables symmetry checks, tail-shape assessment, and simultaneous inference across many $\alpha$ without additional privacy cost. One-shot mechanisms such as \expMech\ and \blb\ target a single pivot, and further analysis requires further privacy budget.

\subsection{Rates, hyper-parameters and privacy budget}\label{sec:rates}

Our accuracy guarantees are asymptotic, and do not provide formal rates (e.g., the coverage is $1-\alpha + f(n)$ for some function $f$), since giving rates is impossible under such minimal distributional assumptions, even in the non-private case. As a result, we cannot compute optimal values of $m$ and $T$ as a function of the sample size and the privacy parameters. The choice of $m$ and $T$ must depend on the optimized quantity (e.g., coverage accuracy, expected width, etc.), the additional distributional assumptions, and the properties of the estimated quantity. 

Here, we set some underlying rules on how to choose the hyperparameters. First, subsampling requires that $m=o(n)$ under minimal distributional assumptions. In fact, in some cases, such as linear statistics, we can choose $m=\theta(n)$ \cite{politis1999subsampling}. A rule of thumb is to consider $m = \omega(\sqrt{n})$. In classic subsampling, $m$ is smaller than $n$ in order to decrease correlation between subsamples. In our case, a smaller $m$ also balances the privacy budget (amplification by subsampling), so it has a second role. As $n \rightarrow \infty$, and the perturbation becomes negligible empirically, we can take larger $m$ (approaching the non-private optimal).
Throughout our experiments, we fixed $m=n^{2/3}$, a number that was derived for non-private subsampling under further assumptions.

Another parameter is the number of subsamples, $T$. Usually, in the non-private literature, $T$ is taken to be ``large enough'' such that the error induced from sampling $T$ instead of $\binom{n}{m}$ subsamples (Monte-Carlo error) is negligible, and in fact, it is treated as zero when analyzing resampling methods. In our case, $T$ plays a crucial role, since we pay a factor of $\sqrt{T}$ or $T$ (advanced or basic composition) in the privacy budget. Since we use the quantiles method, $T$ has to be large enough such that we can take the upper and lower $\alpha/2$ quantile of the empirical distribution. On the other hand, looking at the width of the confidence intervals, once $T$ is not too small, increasing it further does not decrease the expected width, but its variability (see also \ref{app:hyper-parameters}). On the contrary, it only increases the perturbation. Throughout our experiments, we fixed $T=60$, such that $T \cdot m > n$ for simplicity, and did not increase it as a function of $n$. 

Generally, we have that setting $ T\cdot m = \omega(n)$ is not a reasonable choice, since in this case some of the elements are completely discarded, so with these parameters, one might as well split the dataset into $T$ disjoint subsets of size $n/T$. On the other hand, setting $ T\cdot m = O(n^2)$ is also not a reasonable choice, since the perturbation would be too large for most reasonable mechanisms (i.e., noise addition).

Another hyperparameter is the split of the privacy budget between the full-sample estimator and the subsampling estimates. We do not derive a general theory for the optimal split. In Figure \ref{fig:varying_epsilon_split}, we report a sensitivity analysis showing that moderate splits perform similarly, while the best split can be statistic-dependent.

Note that our empirical evaluation considered pure-DP. This has two reasons: first, both \blb\ and \expMech\ are pure-DP, and we wanted the comparison to be as accurate as possible. Second, advanced composition's guarantees are asymptotic. When we chose $\delta=n^{-8}$, in the sample-size regime we chose, advanced composition gave worse $\varepsilon$ than basic composition. For larger sample sizes and $T$, we expect advanced composition to outperform, but for small $n$, we only considered a relatively small $T$, and the asymptotics of basic composition did not kick in yet. That said, for specific mechanisms such as Gaussian noise addition, the composition is very tight and results in very good performances, as we illustrate in Figure~\ref{fig:ks_gaussian_noise}.

\subsubsection{Hyper-parameters under normality setting}\label{app:hyper-parameters}

Consider only cases where the statistic is asymptotically normal, with variance denoted by $\sigma^2$, and the privacy mechanism on each subset is the Gaussian mechanism, so that
$$
\widetilde{\theta}_i = \theta(\dataset(I_i)) + Z_i,
\quad
Z_i \sim N(0,\sigma_z^2), \ \ \Var(\theta(\dataset(I_i))) = \frac{\sigma^2}{m}.
$$

If we denote by $\widetilde{\sigma}^2$ the estimated variance from subsampling, then we have that $\widetilde{\sigma}^2 \approx \frac{\sigma^2}{m}+\sigma_z^2$. If the goal is to minimize the \emph{expected width} of a normal-approximated CI, then increasing $T$ is not beneficial: once $T$ is large enough for the empirical variance to stabilize, increasing it further does not improve the target variance itself, and only increases the privacy perturbation. This supports taking $T$ to be a fixed (or very slowly growing) moderate constant, for example $T=60$. The same logic holds for $m$. Since taking $m<n/T$ is basically neglecting data, we get that choosing $m=n/T$ is optimal for that goal. This is also the result obtained in \cite{dette2025gaussian}. That said, if one chooses $m=n/T$, it is better to do sample-splitting instead of subsampling, so we require that $m=\omega(\sqrt{n})$, on top of $m=o(n)$. In our experiment, we take $m=n^{2/3}$. 

If the goal is \textit{to minimize the mean squared error} of our variance estimate, we end up with very different hyper-parameters regime. Concretely, let
$
\widetilde{\theta}_1,\ldots,\widetilde{\theta}_T
$
be the private estimates computed on the $T$ subsets of size $m$,
and define
$$
\bar{\widetilde{\theta}}
=
\frac{1}{T}\sum_{i=1}^T \widetilde{\theta}_i,
\qquad
\widetilde{\sigma}^2
=
\frac{1}{T-1}\sum_{i=1}^T
\bigl(\widetilde{\theta}_i-\bar{\widetilde{\theta}}\bigr)^2.
$$

Under the additional regularity assumptions of this heuristic discussion, the non-private subset statistic is approximately Gaussian, and its variance across subsets satisfies
$
\frac{m}{n-m}\Var\bigl(\theta(\dataset(I))\bigr) \propto \frac{1}{n}.
$
Therefore, we obtain
$$
\Var\bigl(\widetilde{\sigma}^2\bigr)
=
\widetilde{O}\!\left(
\frac{1}{T}
\left[
\frac{1}{n}
+
\frac{m}{n-m}\sigma_z^2
\right]^2
\right),
$$
where $\widetilde{O}(\cdot)$ hides only logarithmic privacy factors (i.e., $\delta$).

Next, assume that the statistic on a subset of size $m$ has sensitivity of order $1/m$, and that the local privacy budget $\varepsilon_{\operatorname{sub}}$ is related to the total budget $\varepsilon$ through amplification by subsampling and advanced composition, namely
$\varepsilon
\propto
\frac{m}{n}\sqrt{T}\varepsilon_{\operatorname{sub}}.$
Ignoring $\delta$-dependence and logarithmic factors, Gaussian noise addition gives
$
\sigma_z^2
\propto
\frac{1/m^2}{\varepsilon_{\operatorname{sub}}^2}
\propto
\frac{T}{n^2\varepsilon^2}.
$
Substituting this into the previous display and simplifying yields
$$
\Var\bigl(\widetilde{\sigma}^2 \bigr)
=
\widetilde{O}\!\left(
\frac{1}{Tn^2}
+
\frac{m}{(n-m)n^3\varepsilon^2}
+
\frac{m^2T}{(n-m)^2n^4\varepsilon^4}
\right).
$$

This expression highlights the different roles of $T$ and $m$. Increasing $T$ reduces the Monte Carlo fluctuation through the prefactor $1/T$ (first term), but at the same time worsens the privacy noise through composition (last term).

The choice of $m$ is simple in this case. In the minimal-assumption subsampling theory, one requires $m=o(n)$, whereas under the normal limiting distribution it is known that taking $m=\Theta(n)$ is optimal. Optimizing over $T$ gives $T=n\varepsilon^2$, which is exactly the hyper-parameter obtained in \cite{wang2025differentially}. 

We emphasize that this calculation is only a heuristic guide for a more regular regime in which a normal approximation is appropriate. It is not part of our minimal-assumption theory, and does not replace the quantile-based construction analyzed in the main text.

\subsection{Further empirical evaluation of \privBoot}

In the right panel of Figure~\ref{fig:main_fig}, we consider the KS statistic and compare \privSub\ with Laplace noise addition, which satisfies $5$-pure DP, to \privBoot, which uses Gaussian noise addition and satisfies $(\varepsilon=5,\delta=10^{-6})$-DP. We require pure DP for a fair comparison with \blb; however, in this setting, \privBoot\ is additionally allocated a $\delta=10^{-6}$ privacy budget. Allowing \privSub\ the same privacy budget as \privBoot\ enables substantially improved performance. In this case, we compute the exact privacy loss distribution (PLD) for Gaussian noise via numerical evaluation, rather than relying on the general analytic bound of Theorem~\ref{thm:eps_delta_total}. We note that the privacy analysis of \privBoot\ is likewise performed numerically. Accordingly, we include a Gaussian-noise variant of our method, denoted \texttt{PrivSub(Gaussian)}, and compare it to the pure-DP Laplace-noise variant shown in Figure~\ref{fig:main_fig}, denoted \texttt{PrivSub(Laplace)}.

Figure~\ref{fig:ks_gaussian_noise} demonstrates that, under a fair comparison in which both \privBoot\ and \privSub\ are assigned the same privacy budget $(\varepsilon=5,\delta=10^{-6})$, \privSub\ substantially outperforms \privBoot, achieving performance close to the non-private baseline, with only minor under-coverage at small sample sizes.

\begin{figure*}[t!]
\vspace{.3in}
\centering\includegraphics[width=0.8\linewidth]{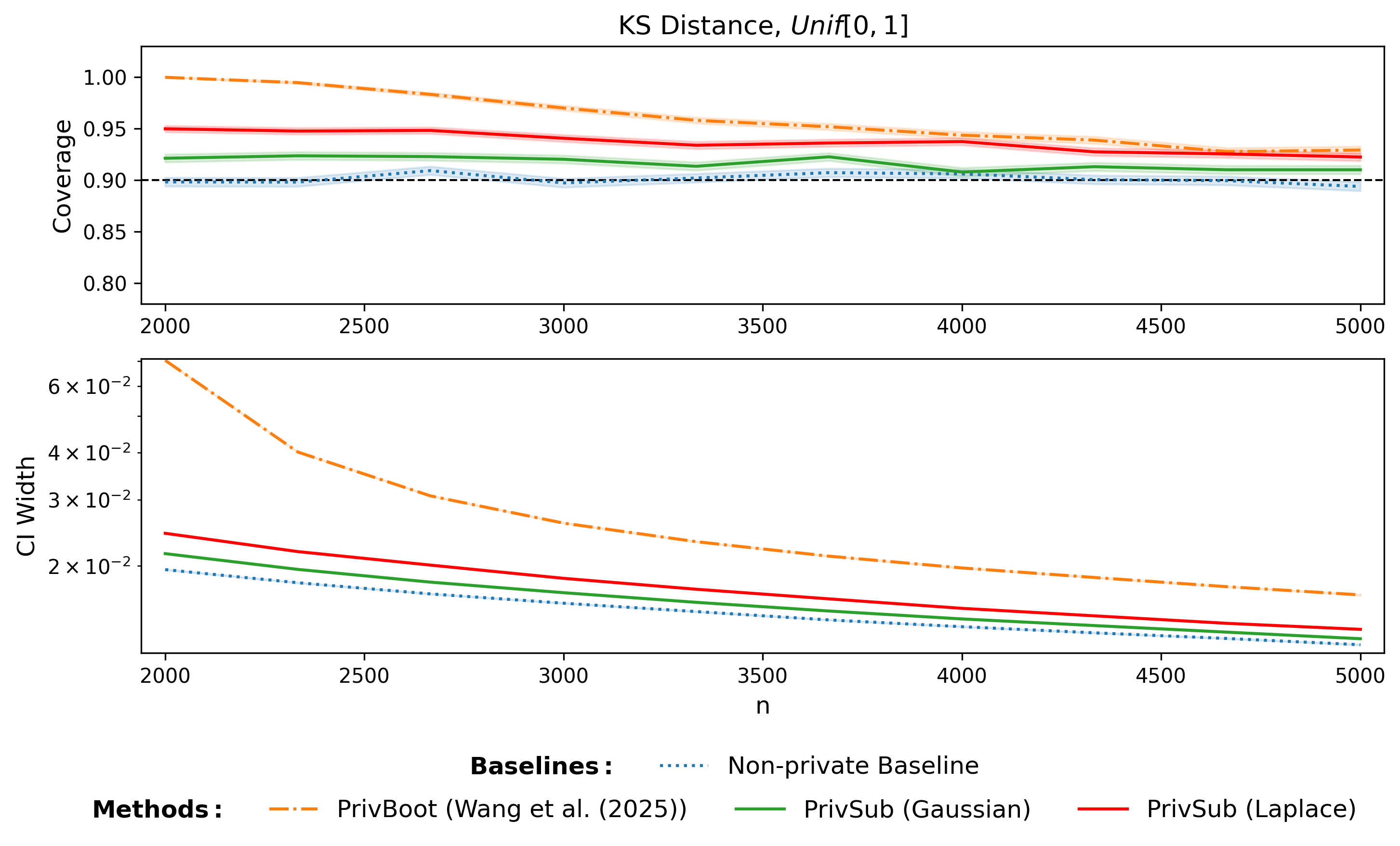}
    \caption{Comparison of \privSub\ and \privBoot\ for confidence interval construction of the KS statistic under matched privacy budgets. Both methods are assigned the same privacy budget, $(\varepsilon=5,\delta=10^{-6})$-DP, and use Gaussian noise addition. We include two variants of our method: \texttt{PrivSub (Gaussian)}, which uses Gaussian noise, and \texttt{PrivSub (Laplace)}, which uses Laplace noise and satisfies $5$-pure DP (shown in Figure~\ref{fig:main_fig}). The results demonstrate that, under a fair privacy comparison, \privSub\ substantially outperforms \privBoot, achieving performance close to the non-private baseline, with only minor under-coverage at small sample sizes.}

\label{fig:ks_gaussian_noise}
\end{figure*}

\subsection{Full empirical evaluation}\label{subsec:full_emp_eval}

In this Appendix, we test the robustness of our results by varying the privacy budget (Figures~\ref{fig:median_comparison_normal},
\ref{fig:logistic_comparison}, and \ref{fig:ks_comparison_uniform}), the subset size, $m$ (Figure~\ref{fig:varying_m}) the number of subsets (Figure~\ref{fig:varying_T}), the split of the privacy budget (Figure~\ref{fig:varying_epsilon_split}), and the confidence level (Figures ~\ref{fig:varying_significance_width} and \ref{fig:varying_significance_coverage}). 
Overall, our results show that the method is robust to hyperparameter choices and can be improved by calibrating them to the specific problem.

\paragraph{Varying the privacy budget}
We repeat the experiment described in Section \ref{sec:num_study} varying the privacy parameter, $\varepsilon=2,5,8$. Overall, Figures~\ref{fig:median_comparison_normal}, \ref{fig:logistic_comparison}, and 
\ref{fig:ks_comparison_uniform} demonstrate the robustness of our results across a range of privacy parameters. In Figure~\ref{fig:median_comparison_normal}, we observe that for $\varepsilon=2$, \privSub\ under-performs relative to \privBoot; however, this gap narrows as the sample size increases. This behavior is attributable to an overestimation of variability induced by rescaling the private empirical CDF from subsamples of size~$m$ to the full sample size~$n$. In contrast, when $\varepsilon=8$, \privSub\ performs comparably to the non-private baseline, while \privBoot\ continues to exhibit a substantial gap from the non-private benchmark. \blb\ shows a similar trend for $\varepsilon=8$, with some under-coverage. 

In Figure~\ref{fig:logistic_comparison}, we consider the logistic-regression slope experiment. For a small privacy budget, $\varepsilon=2$, \privSub\ produces wider intervals and noticeable under-coverage at smaller sample sizes, but its coverage improves with $n$, and its width decreases toward the non-private baseline. For $\varepsilon=5$ and $\varepsilon=8$, \privSub\ is well calibrated and remains close to the non-private benchmark in both coverage and width. In contrast, \privBoot\ is substantially more conservative and produces much wider intervals for $\varepsilon=2,5$, while \blb\ often has competitive widths with good coverage.

For interpreting the \blb\ results, we note that one of its hyperparameters is the range over which the CI endpoints are searched. Following \citet{chadha2024resampling}, we use a range whose width is about $2.5$ times the true CI width. This is a favorable tuning choice, since it depends on the true variance of the estimator rather than only on the problem domain. We keep this choice to match the original implementation, but a more conservative range can substantially affect the performance of \blb.

In Figure~\ref{fig:ks_comparison_uniform}, we consider the KS statistic for data drawn from the $\mathrm{Unif}[0,1]$ distribution. The figure illustrates that \privSub\ extends naturally to non-smooth statistics with non-normal limiting distributions, achieving performance close to the non-private baseline. \privSub\ is well calibrated and performs favorably, while \privBoot\ yields substantially wider confidence intervals and, for $\varepsilon=8$, simultaneously exhibits under-coverage. As in Figure \ref{fig:main_fig}, we can see that \blb\ produces substantially wider confidence intervals.

\begin{figure}[ht]
\vspace{.3in}
\centering\includegraphics[width=\linewidth]{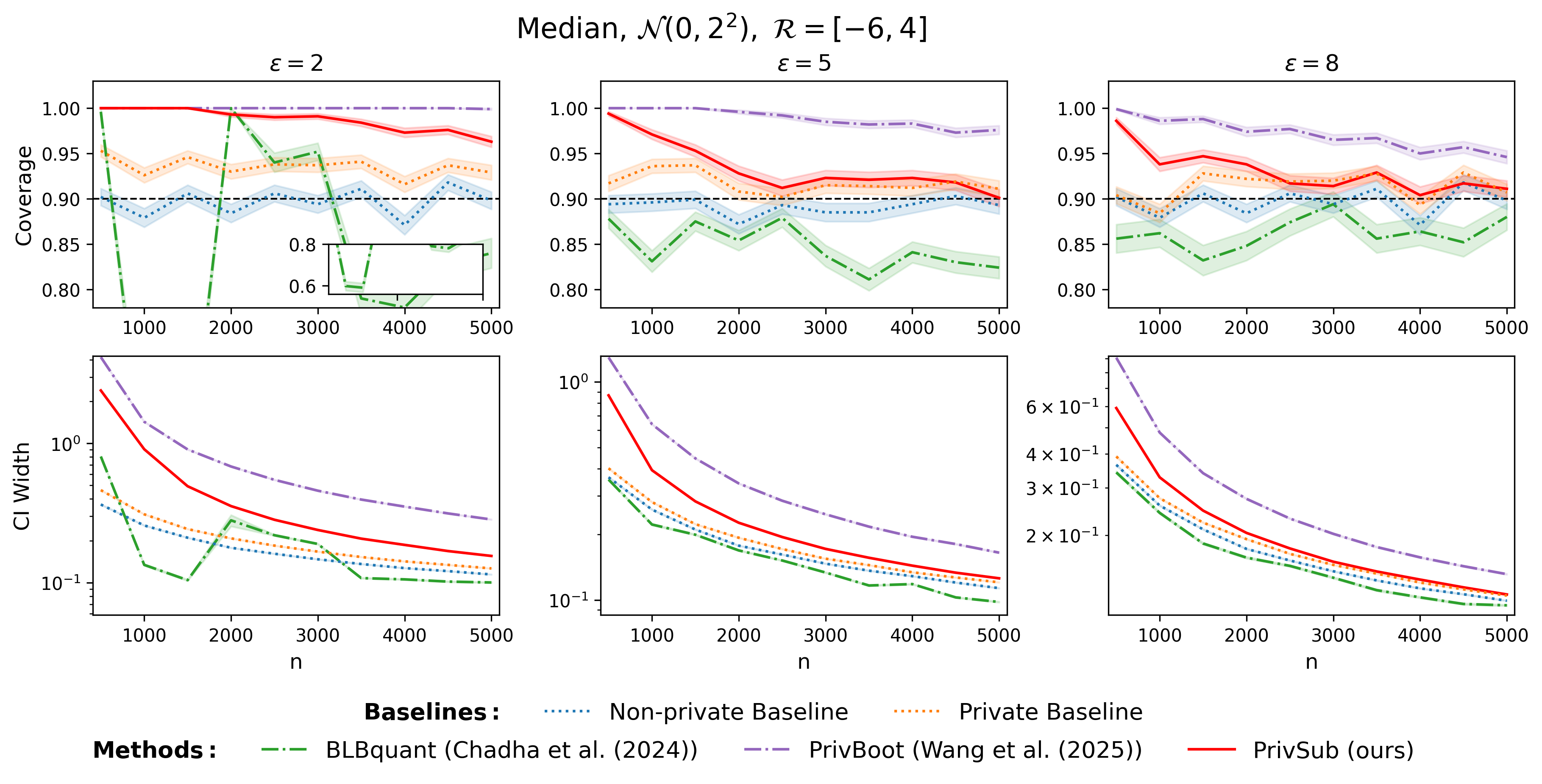}
\caption{
Comparison of our method (\privSub) with existing general, nonparametric DP CI methods: the BLB-based approach (\blb; \citealp{chadha2024resampling}) and the private bootstrap (\privBoot; \citealp{wang2025differentially}). We also include two baselines: a median-specific private method based on the exponential mechanism (\expMech; \citealp{drechsler2022nonparametric}), and non-private procedures—bootstrap for median estimation referred to as the non-private baseline. We evaluate $0.9$-CI estimation for the median under a truncated normal distribution. All methods operate under $\varepsilon$-pure DP ($\varepsilon=2,5,8$), except for \privBoot, which additionally uses a privacy parameter of $\delta=10^{-6}$. A detailed discussion is provided in Section~\ref{sec:num_study}.}
\label{fig:median_comparison_normal}
\end{figure}

\begin{figure}[ht]
\vspace{.3in}
\centering\includegraphics[width=\linewidth]{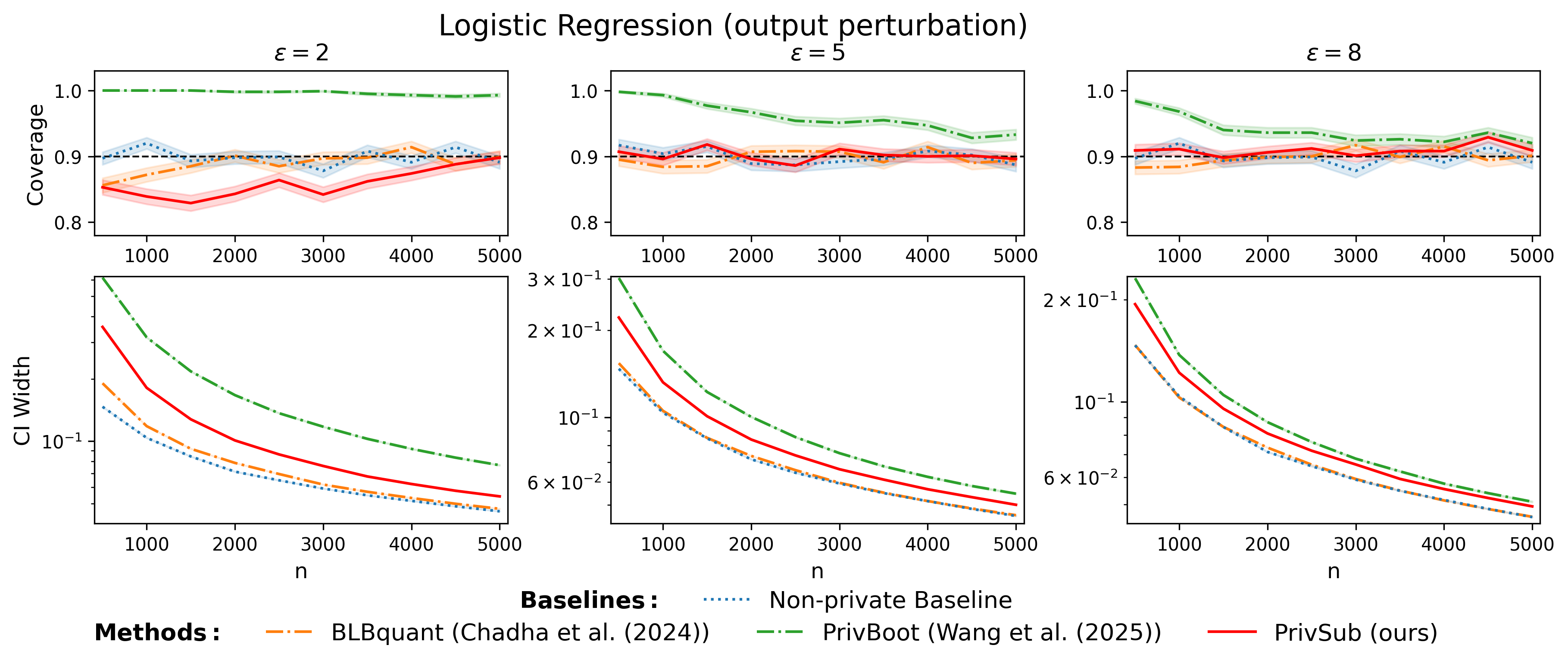}
\caption{Comparison of our method (\privSub) with existing general, nonparametric DP CI methods: the BLB-based approach (\blb; \citealp{chadha2024resampling}) and the private bootstrap (\privBoot; \citealp{wang2025differentially}). We also include a non-private bootstrap referred to as the non-private baseline. We evaluate $0.9$-CI estimation of logistic regression slope coefficient, where the data are drawn from the uniform distribution, and the response from the logistic model. All methods operate under $(\varepsilon,\delta=10^{-6})$-DP, ($\varepsilon=2,5,8$). A detailed discussion is provided in Section~\ref{sec:num_study}.}
\label{fig:logistic_comparison}
\end{figure}

\begin{figure}[ht]
\vspace{.3in}
\centering\includegraphics[width=\linewidth]{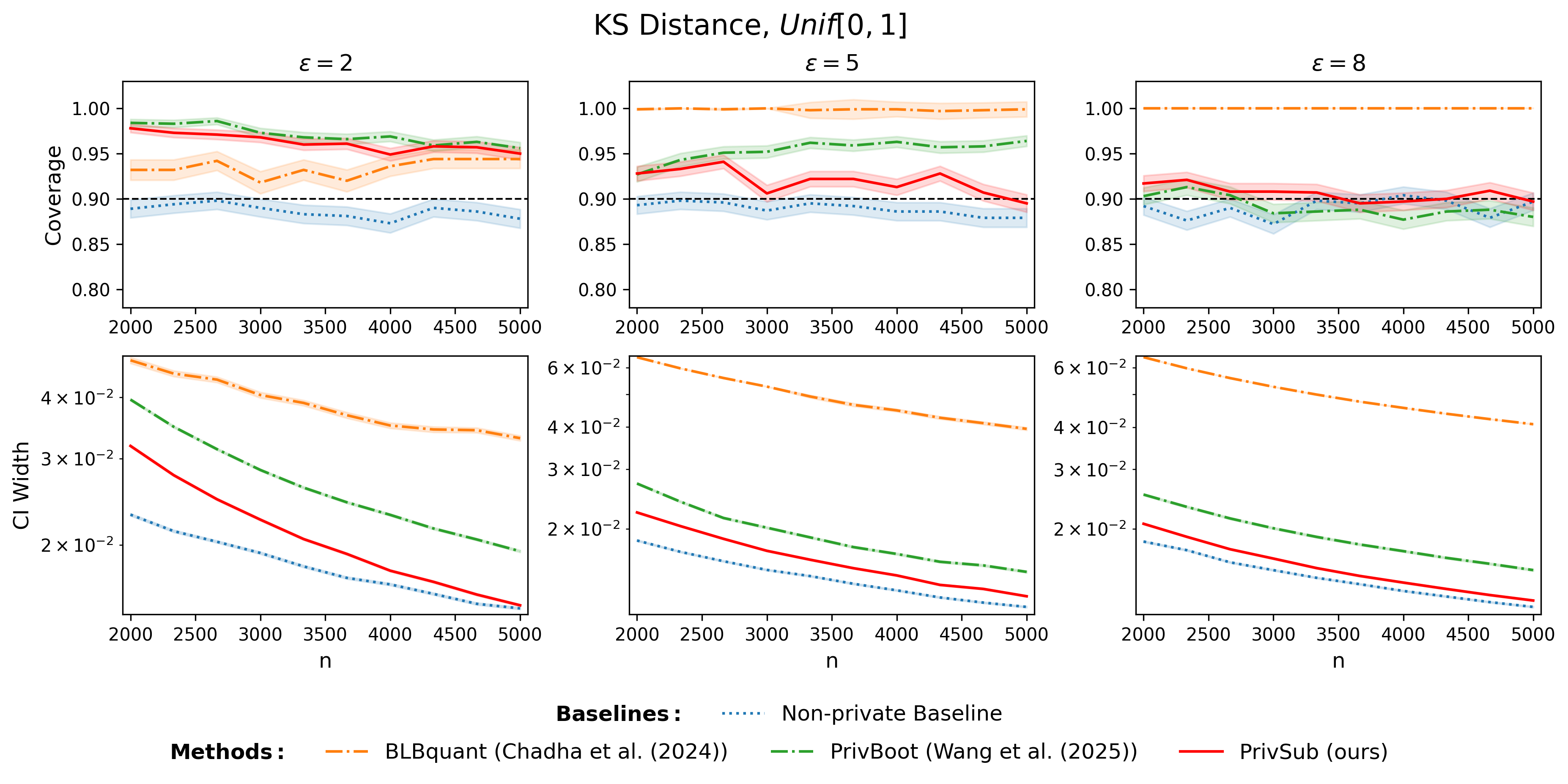}
\caption{
Comparison of our method (\privSub) with existing general, nonparametric DP CI methods: the BLB-based approach (\blb; \citealp{chadha2024resampling}) and the private bootstrap (\privBoot; \citealp{wang2025differentially}). We also include non-private subsampling for the KS statistic referred to as the non-private baseline. We evaluate $0.9$-CI estimation KS distance in uniformity testing, where the data are drawn from the uniform distribution. All methods operate under $\varepsilon$-pure DP ($\varepsilon=2,5,8$), except for \privBoot, which additionally uses a privacy parameter of $\delta=10^{-6}$. A detailed discussion is provided in Section~\ref{sec:num_study}.}
\label{fig:ks_comparison_uniform}
\end{figure}

\paragraph{Varying the subset size, $m$.}
Figure~\ref{fig:varying_m} studies the sensitivity of \privSub\ to the choice of the subsample size $m$. Across all three settings, the method remains close to nominal coverage for a fairly wide range of choices of $m$, suggesting that the procedure is not overly sensitive to this tuning parameter. As expected, $m=\sqrt{n}$ provides the narrower confidence intervals, as explained in Section \ref{sec:priv_and_acc_privsub}. Increasing $m$ reduces the privacy amplification obtained from subsampling, and hence requires more noise in each private release. The empirical results suggest that, in the regimes considered here, the statistical benefit of increasing $m$ slightly is outweighed by this additional privacy cost, leading to modestly smaller widths for smaller choices of $m$. Nevertheless, the differences are not dramatic, and the default choice $m=n^{2/3}$ provides a stable middle ground: it achieves near-nominal coverage while producing widths close to those obtained by the more aggressive choices of $m$.

\begin{figure}[ht]
\vspace{.3in}
\centering\includegraphics[width=\linewidth]{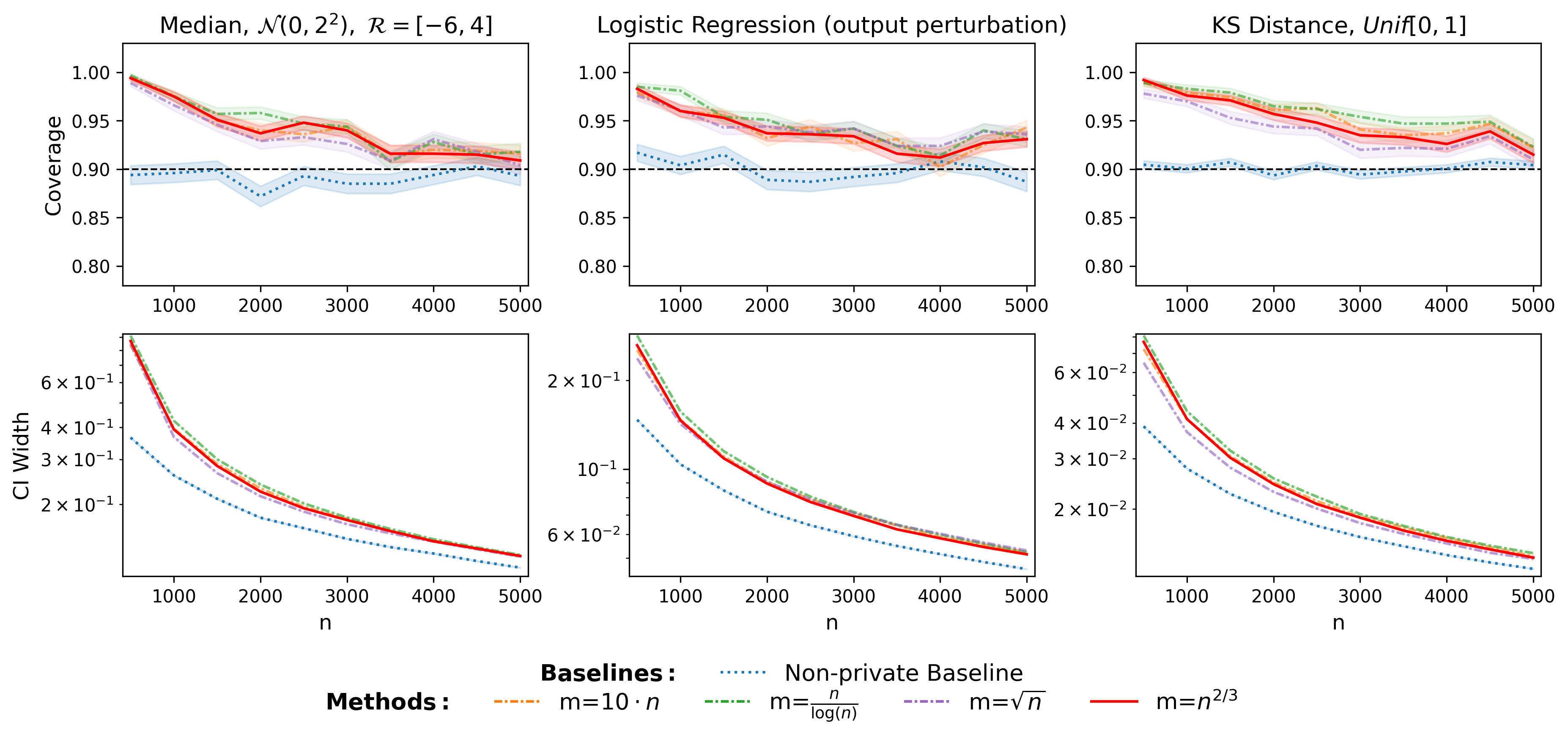}
\caption{Sensitivity of \privSub\ to the subsample-size function $m(n)$.\\
Finite-sample performance of \privSub\ under different choices of the subsample size $m$ as a function of the full sample size $n$. Results are shown for median estimation, logistic-regression slope estimation with output perturbation, and KS-distance estimation, with $T=60$ and $\varepsilon=5$. The tested choices are $m=0.1n$, $m=n/\log(n)$, $m=
n^{2/3}$, and $m=\sqrt n$, with $m=n^{2/3}$ used throughout the main experiments. The dotted curve denotes the corresponding non-private baseline: bootstrap for the median and logistic-regression experiments, and subsampling for the KS statistic. For width, shaded bands indicate $\pm2$ empirical standard errors across repetitions; for coverage, shaded bands indicate $\pm 2$ binomial standard errors. 
Across these choices, \privSub\ maintains near-nominal $0.9$ coverage, while larger subsample sizes generally yield slightly wider confidence intervals. \\
In general, smaller choices of $m$ are preferable, although the optimal choice is statistic-dependent; nevertheless, performance remains relatively robust across the tested $m(n)$ functions.}
\label{fig:varying_m}
\end{figure}

\paragraph{Varying the number of subsets, $T$.}
Figure~\ref{fig:varying_T} examines the sensitivity of \privSub\ to the number of subsamples $T$. Overall, the results indicate that the method is fairly robust to this choice: across the three examples, all values of $T$ considered yield near-nominal coverage and comparable confidence interval widths. At the same time, the figure suggests that taking $T$ unnecessarily large is not always beneficial. Increasing $T$ provides a more stable empirical approximation to the subsampling distribution, but it also increases the privacy cost through composition, and hence the magnitude of the privacy perturbation. In several of the settings considered here, smaller or moderate values of $T$ therefore achieve similar, and sometimes slightly better, performance than larger values. This supports using a relatively modest value of $T$ in practice, once the empirical distribution is estimated with sufficient stability. The optimal choice, however, is not uniform across statistics: depending on the stability of the statistic and the relative contribution of privacy noise, larger values of $T$ may still be useful. In our experiments, the default choice $T=60$ provides a conservative and stable compromise across all settings.

\begin{figure}[ht]
\vspace{.3in}
\centering\includegraphics[width=\linewidth]{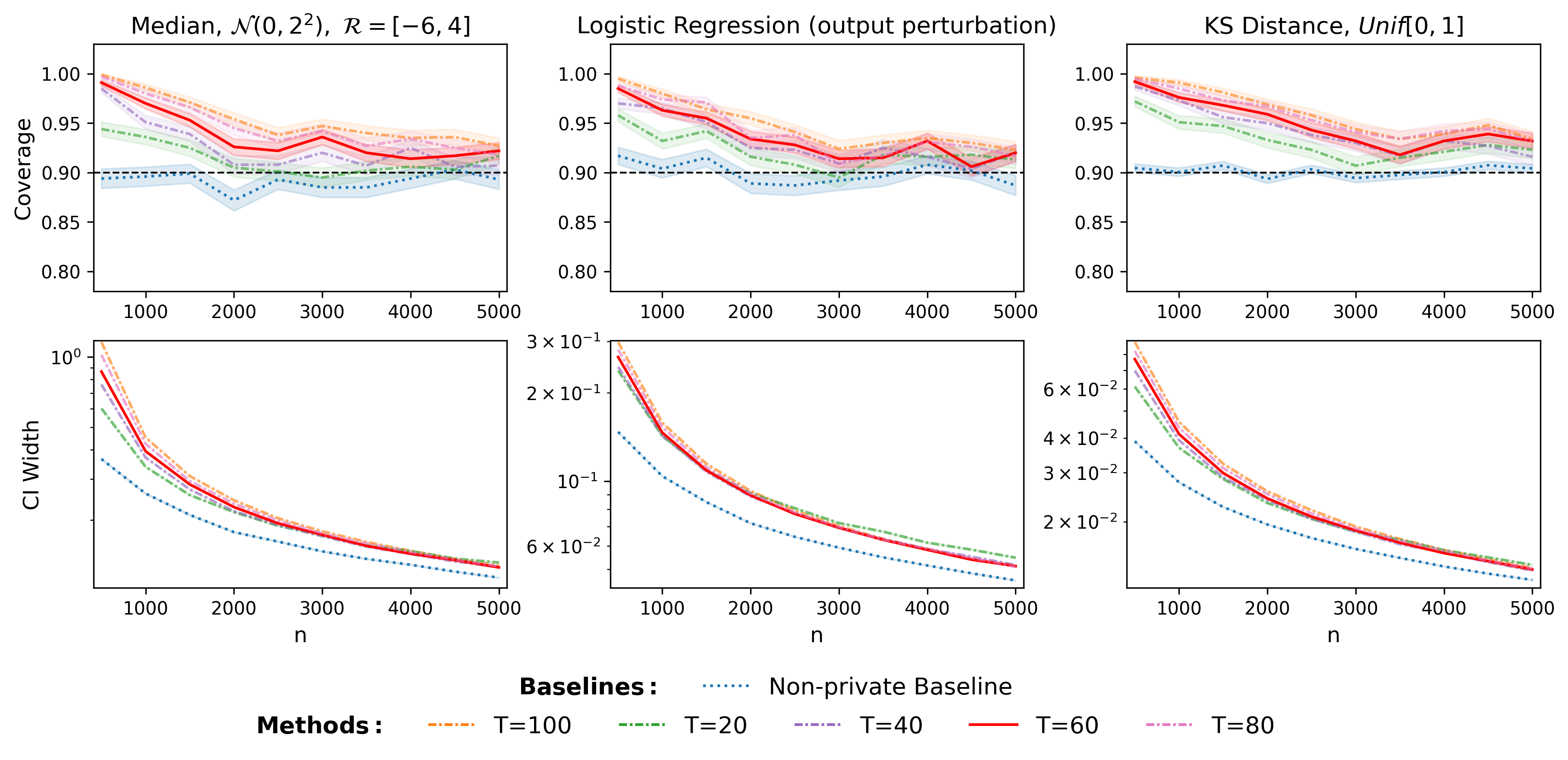}
\caption{Sensitivity of \privSub\ to the number of subsamples $T$.\\
Finite-sample performance of \privSub\ under different choices of the number of subsamples $T$ used to construct the private confidence interval. Results are shown for median estimation, logistic-regression slope estimation with output perturbation, and KS-distance estimation, for $m=n^{2/3}$ and $\varepsilon=5$. The dotted curve denotes the corresponding non-private baseline: bootstrap for the median and logistic-regression experiments, and subsampling (with $m=n^{2/3}$) for the KS statistic. For width, shaded bands indicate $\pm 2$ empirical standard errors across repetitions; for coverage, shaded bands indicate $\pm 2$ binomial standard errors. The curve corresponding to \(T=60\) indicates the fixed default number of subsamples used throughout the main experiments; it is included here as a reference rather than as a tuned choice.
Across the tested range of subsample counts, \privSub\ maintains near-nominal $0.9$ coverage while achieving comparable CI widths, indicating that performance is relatively robust to this tuning parameter.\\
$T$ should generally be kept relatively small, though the optimal choice is statistic-dependent: larger $T$ increases the composition cost, forcing more noise per release under a fixed total privacy budget, whereas overly small $T$ can yield overly conservative quantiles.}
\label{fig:varying_T}
\end{figure}

\paragraph{Varying the privacy budget split.}
Figure~\ref{fig:varying_epsilon_split} studies the effect of splitting the privacy budget between the full-sample estimate and the subsampling step. We let $q$ denote the fraction of the total budget allocated to the full-sample statistic, so that $(\varepsilon_{\mathrm{full}},\delta_{\mathrm{full}}) = q(\varepsilon,\delta)$, while the remaining budget is used for the private subsample estimates. This split controls a natural tradeoff: allocating more budget to the full-sample statistic improves the accuracy of the center of the CI, but leaves less budget for estimating the subsampling distribution, thereby increasing the privacy perturbation in the CI calibration step. 

Overall, the procedure is fairly stable for moderate choices of the split. In the median and KS experiments, the choices $q \in \{0.25,0.4,0.5,0.6\}$ lead to broadly similar coverage and width, while $q=0.75$ is noticeably more conservative and yields wider intervals. This suggests that allocating too little budget to the subsampling stage can be harmful, since the empirical distribution used to calibrate the CI becomes overly noisy. The logistic-regression experiment shows that the best split can be statistic-dependent: in this setting, $q=0.25$ gives particularly well-calibrated coverage and relatively narrow intervals. These results suggest that the optimal privacy split depends on the statistic, mechanism, and sample size. For simplicity and to avoid tuning the split separately in each experiment, we use the balanced choice $q=0.5$ as our default.

\begin{figure}[ht]
\vspace{.3in}
\centering\includegraphics[width=\linewidth]{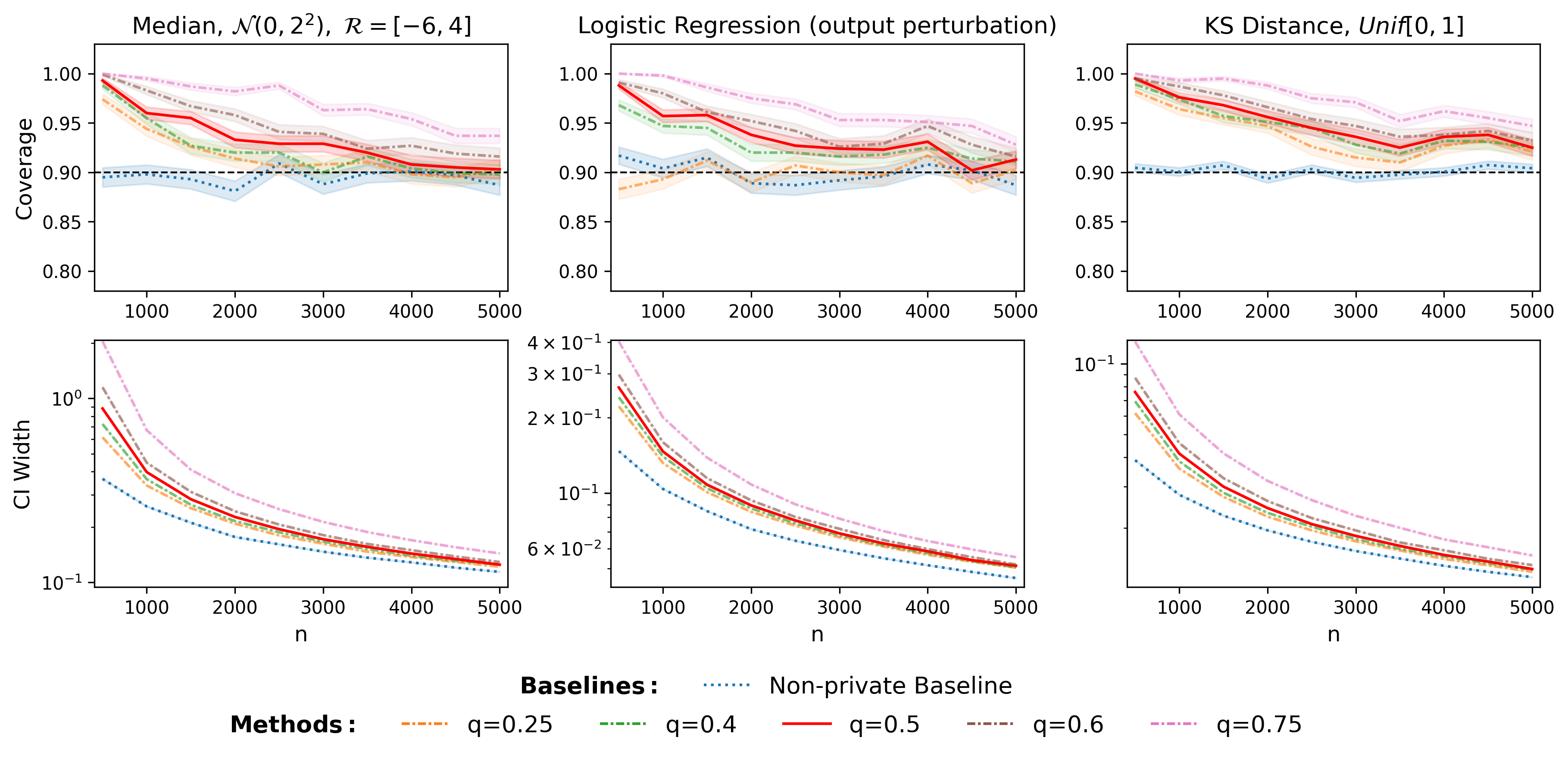}
\caption{Sensitivity of \privSub\ to the privacy-budget split parameter $q$.\\
Finite-sample performance of \privSub\ under different allocations of the total privacy budget $\varepsilon=5$ between center estimation and subsampling, where $q\varepsilon$ is assigned to $\varepsilon_{\mathrm{full}}$ and $(1-q)\varepsilon$ is assigned to $\varepsilon_{\mathrm{sub}}$.  The dotted curve denotes the corresponding non-private baseline: bootstrap for the median and logistic-regression experiments, and subsampling with $T=60, m=n^{2/3}$ for the KS statistic. The red curve corresponds to the fixed default allocation \(q=0.5\) used throughout the main experiments; it is included here as a reference rather than as a tuned choice. For width, shaded bands indicate $\pm 2$ empirical standard errors across repetitions; for coverage, shaded bands indicate $\pm 2$ binomial standard errors.\\
Results are shown for median estimation, logistic-regression slope estimation with output perturbation, and KS-distance estimation. Across a range of values $q \in \{0.25,0.4,0.5,0.6,0.75\}$, \privSub\ maintains near-nominal $0.9$ coverage while achieving comparable CI widths, indicating that performance is relatively robust to the choice of privacy-budget split.
}
\label{fig:varying_epsilon_split}
\end{figure}

\paragraph{Varying the confidence level, $1-\alpha$}
Figures~\ref{fig:varying_significance_width} and~\ref{fig:varying_significance_coverage}
compare the widths and coverages of CIs across nominal confidence levels
$1-\alpha \in \{0.85,0.9,0.95\}$.
As expected, increasing the confidence level leads to wider intervals for all methods and all
statistics. The qualitative comparison between the methods is stable across the three choices of
$\alpha$: \privSub\ consistently produces substantially tighter intervals than \privBoot, with the largest gains at smaller sample sizes and in the KS experiment. For median estimation and logistic regression, \privSub\ remains close to the private and non-private baselines in terms of width as $n$ grows, while for the KS statistic it is much closer to the non-private subsampling baseline than \privBoot.

The corresponding coverage results show that this reduction in width does not come from systematic
under-coverage. Across the three nominal levels, \privSub\ generally tracks the target coverage and
improves as $n$ increases, with mild finite-sample under-coverage in some of the median and logistic regression settings at smaller sample sizes. In contrast, \privBoot\ is consistently conservative, especially for the median and KS statistic, which explains its substantially larger widths. Overall, these experiments indicate that the main conclusions are robust to the choice of confidence level:
changing $\alpha$ affects the absolute width and the nominal coverage target, but does not change the relative behavior of the methods.

For the KS statistic, the mild under-coverage of the non-private baseline is not an implementation artifact, but rather a finite-sample statistical feature of this particular non-smooth statistic; since it is orthogonal to the privacy question studied here, we do not pursue this issue further.

\begin{figure}[ht]
\vspace{.3in}
\centering\includegraphics[width=\linewidth]{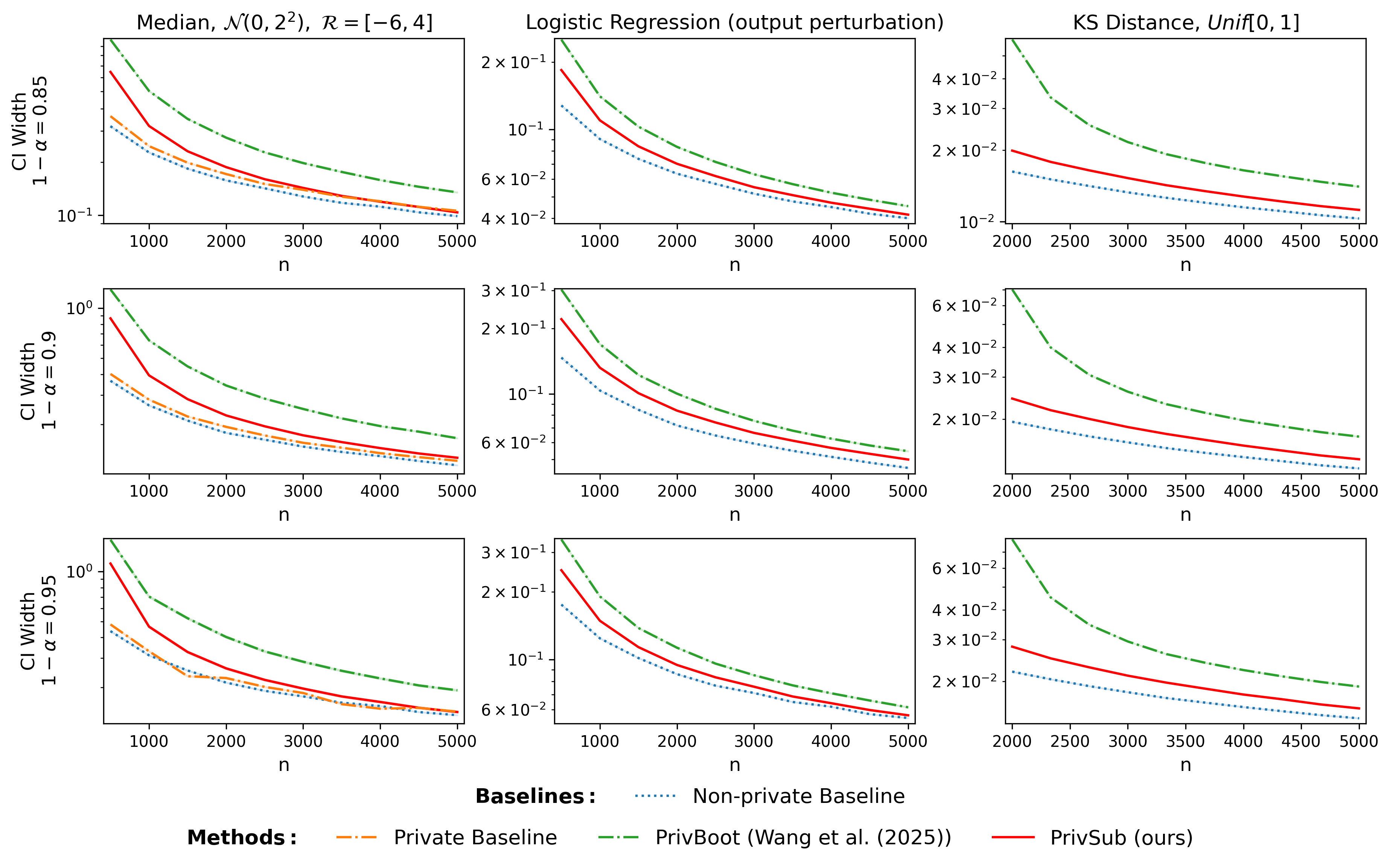}
\caption{Sensitivity of private CI width to the nominal coverage level \(\alpha\).\\
Comparison of our method (\privSub) with existing general, nonparametric DP CI methods: the private bootstrap (\privBoot; \citealp{wang2025differentially}), a median-specific private method based on the exponential mechanism (\expMech; \citealp{drechsler2022nonparametric}), alongside  non-private baseline: bootstrap for the median and
logistic-regression experiments, and subsampling for the KS statistic. We evaluate finite-sample CI width as a function of sample size $n$ across three statistical tasks: median estimation under $\mathcal N(0,2^2)$ with $\mathcal{R}=[-6,4]$, logistic-regression slope estimation with output perturbation, and KS-distance estimation under $\operatorname{Unif}[0,1]$. For the median and KS experiments, all private methods are $5$-pure
DP, except \privBoot, which is ($\varepsilon = 5, \delta = 10^{-6}$)-DP; for logistic regression, all private methods
are ($\varepsilon = 5, \delta = 10^{-6}$)-DP.
Results are shown for nominal coverage levels $1-\alpha\in\{0.85,0.9,0.95\}$.\\
Across all tasks and coverage levels, CI widths decrease with increasing $n$. \privSub\ consistently yields substantially narrower intervals than \privBoot, with the largest gains at smaller sample sizes, and remains close to the private and non-private baselines as $n$ grows, indicating robust improvement across confidence levels.
}
\label{fig:varying_significance_width}
\end{figure}

\begin{figure}[ht]
\vspace{.3in}
\centering\includegraphics[width=\linewidth]{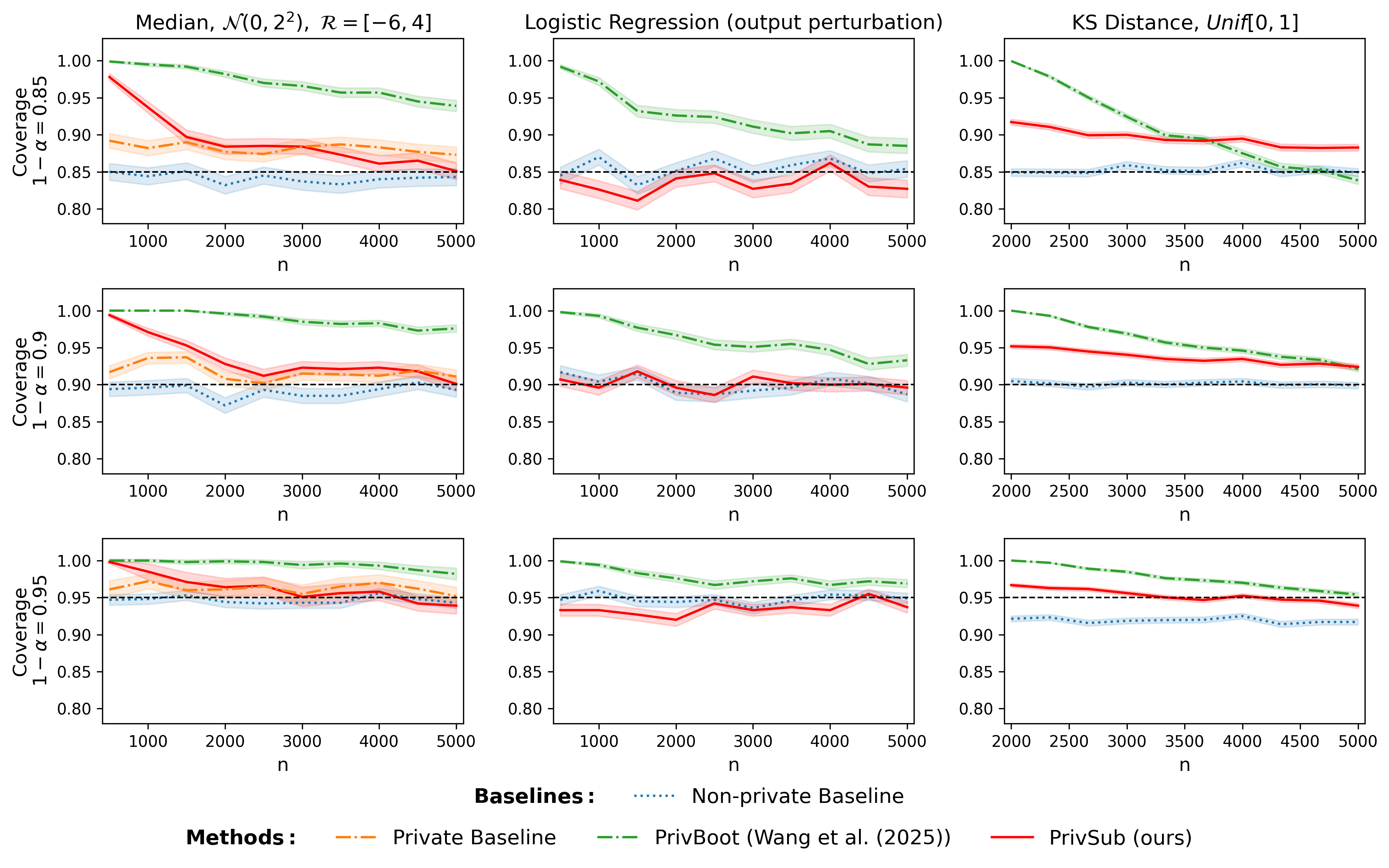}
\caption{Sensitivity of private CI coverage to the nominal coverage level \(\alpha\).\\
Comparison of our method (\privSub) with existing general, nonparametric DP CI methods: the private bootstrap (\privBoot; \citealp{wang2025differentially}), a median-specific private method based on the exponential mechanism (\expMech; \citealp{drechsler2022nonparametric}), alongside the non-private baseline: bootstrap for the median and
logistic-regression experiments, and subsampling for the KS statistic. We evaluate finite-sample CI coverage as a function of sample size $n$ across three statistical tasks: median estimation under $\mathcal N(0,2^2)$ with $\mathcal{R}=[-6,4]$, logistic-regression slope estimation with output perturbation, and KS-distance estimation under $\operatorname{Unif}[0,1]$. For the median and KS experiments, all private methods are $5$-pure
DP, except \privBoot, which is ($\varepsilon = 5, \delta = 10^{-6}$)-DP; for logistic regression, all private methods
are ($\varepsilon = 5, \delta = 10^{-6}$)-DP.
Results are shown for nominal coverage levels $1-\alpha\in\{0.85,0.9,0.95\}$. Shaded bands indicate $\pm 2$ binomial standard errors across repetitions.\\
Across all tasks and coverage levels, \privSub\ achieves stable empirical coverage close to the target level and is generally comparable to the private and non-private baselines. In contrast, \privBoot\ tends to be conservative, especially at lower nominal coverage levels, while remaining above the target in many settings. The same qualitative pattern holds across $\alpha$, 
indicating that \privSub\ provides robust finite-sample coverage across confidence levels and statistical tasks.}
\label{fig:varying_significance_coverage}
\end{figure}

\subsection{Comparison of CDF estimation}

Our method of proof of the validity goes through proving point-wise (or uniform) convergence of the distribution (Theorem \ref{thm:cons_priv_quan_CI}). In this subsection, we show the convergence of the CDF as a function of the sample size. Figure ~\ref{fig:cdf} provides empirical verification of this convergence for two representative distributions. Unlike Figure 1, we use the smaller privacy budget $\varepsilon=2$ to make the effect of privacy noise visible; as $n$ increases, the private CDF nevertheless approaches the theoretical target.

The figure displays the estimated CDF from \privSub~(solid blue line) alongside the non-private empirical CDF (dashed orange line) and the theoretical CDF (dashed black line) for sample sizes  $n=1,000, 3,000, 5,000$. 

Several patterns emerge: First, as sample size increases, both the private and non-private CDFs converge toward the theoretical CDF, with the gap narrowing systematically. We can clearly see that the distribution estimated by \privSub\ is less concentrated when the underlying distribution is normal, consistent with our previous empirical evaluations. 

\begin{figure}[ht]
\vspace{.3in}
\centering\includegraphics[width=\linewidth]{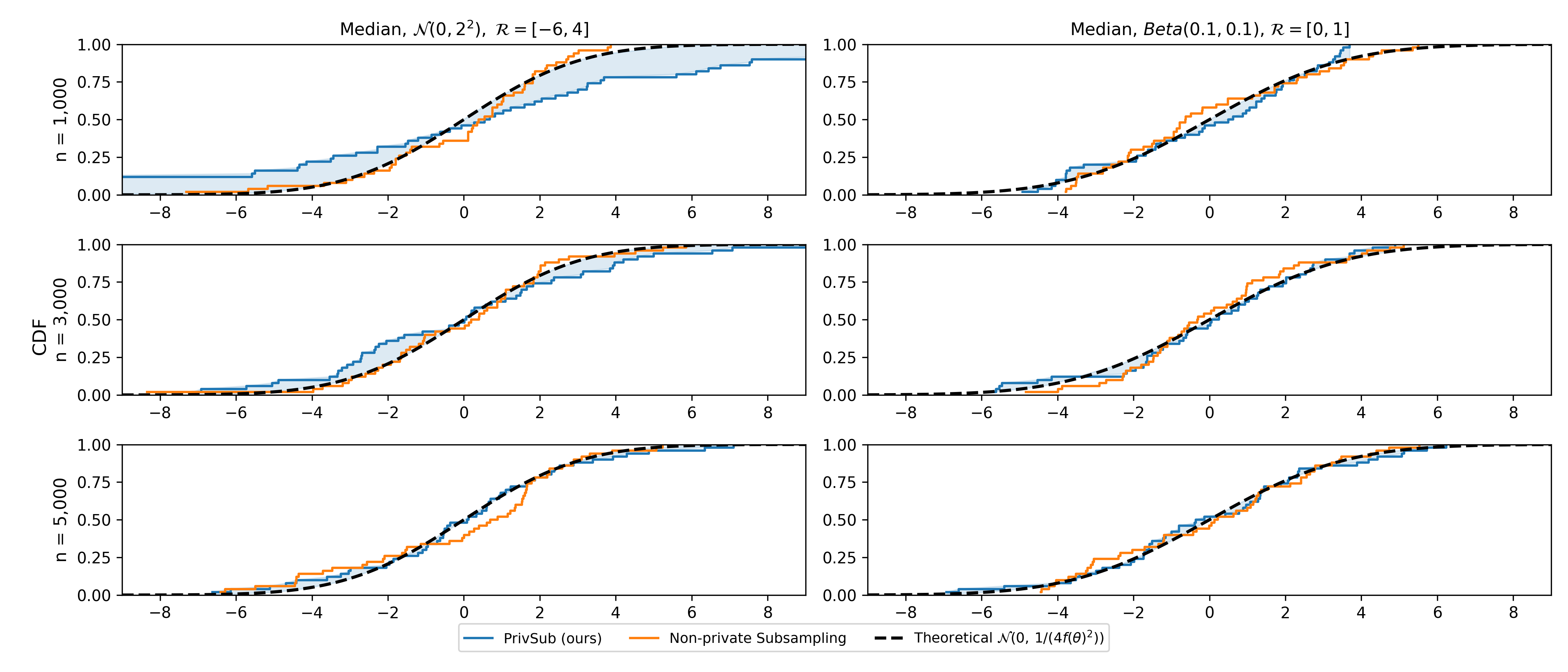}
\caption{Empirical CDFs of the median under different distributions with $\varepsilon=2$. Three methods are shown: the theoretical centralized estimator (with its theoretical CDF), non-private subsampling, and \privSub. Both subsampling methods (\privSub\ and {non-private subsampling}) partition the dataset to subsamples of size $n^{2/3}$, repeat this process $T=50$ times, and construct empirical CDFs from the resulting estimates.}
\label{fig:cdf}
\end{figure}


\end{document}